\documentclass[%
 aip,
 reprint,%
]{revtex4-1}

\usepackage{graphicx}
\usepackage{dcolumn}
\usepackage{bm}
\usepackage{hyperref}
\usepackage[utf8]{inputenc}
\usepackage[T1]{fontenc}
\usepackage{mathptmx}
\usepackage{enumerate}
\usepackage{booktabs}
\usepackage{tabularx}
\usepackage{multirow}
\usepackage{lipsum}
\usepackage{microtype}

\usepackage[sc]{mathpazo}
\usepackage[mathscr]{eucal}

\DeclareMathAlphabet{\mathpzc}{OT1}{pzc}{m}{it}
\newcommand\aH{{\mathpzc H}}
\newcommand\aG{{\mathpzc G}}
\newcommand\aL{{\mathpzc L}}
\newcommand\ah{{\mathpzc h}}
\newcommand\ag{{\mathpzc g}}
\newcommand\al{{\mathpzc l}}
\usepackage{amsmath,amssymb,amsthm}
\usepackage{mathtools}
\usepackage{subfigure}
\usepackage[english]{babel}
\usepackage{lmodern}
\newtheorem{theorem}{Theorem}

\newtheorem{lemma}{Lemma}
\newtheorem{remark}{Remark}

\newcommand\mvector{\boldsymbol}
\newcommand\vb{\mvector{b}}

\newcommand\ve{\mvector{e}}
\newcommand\vl{\mvector{l}}
\newcommand\vm{\mvector{m}}

\newcommand\vq{\mvector{q}}
\newcommand\vr{\mvector{r}}

\newcommand\vx{\mvector{x}}
\newcommand\vA{\mvector{A}}
\newcommand\vB{\mvector{B}}

\newcommand\vE{\mvector{E}}

\newcommand\vI{\mvector{I}}
\newcommand\vJ{\mvector{J}}

\newcommand\vL{\mvector{L}}
\newcommand\vM{\mvector{M}}

\newcommand\vR{\mvector{R}}
\newcommand\vS{\mvector{S}}
\newcommand\vV{\mvector{V}}
\newcommand\vX{\mvector{X}}
\newcommand\vZ{\mvector{Z}}

\newcommand\vz{\mvector{z}}
\newcommand\valpha{\mvector{\alpha}}
\newcommand\vgamma{\mvector{\gamma}}

\newcommand\vomega{\mvector{\omega}}
\newcommand\vOmega{\mvector{\Omega}}
\newcommand\cM{{\mathcal M}}
\newcommand\cN{{\mathcal N}}

\newcommand\scA{{\mathscr A}}
\newcommand\scD{{\mathscr D}}

\newcommand\scG{{\mathscr G}}

\newcommand\scS{{\mathscr S}}
\newcommand\scT{{\mathscr T}}
\newcommand\field{\mathbb}

\newcommand\R{\field{R}}

\newcommand\C{\field{C}}
\newcommand\J{\field{J}}
\newcommand\Z{\field{Z}}

\newcommand\N{\field{N}}

\newcommand\vGamma{\mvector{\Gamma}}
\newcommand\vzero{\mvector{0}}
\renewcommand\Re{\operatorname{Re}}
\renewcommand\Im{\operatorname{Im}}

\newcommand\rmd{\mathrm{d}}
\newcommand\rme{\mathrm{e}}
\newcommand\rmi{\mathrm{i}}

\newcommand\Dtt{\frac{\mathrm{d}^2\phantom{t} }{\mathrm{d}t^2}}
\newcommand\pder[2]{\dfrac{\partial #1 }{\partial #2}}
\newcommand\Dz{\frac{\mathrm{d}\phantom{z} }{ \mathrm{d}z}}
\newcommand\ord{\operatorname{ord}}
\newcommand\mtext[1]{\quad\text{#1}\quad}
\newcommand\abs[1]{\lvert #1 \rvert}
\newcommand\id{\operatorname{\mathrm{Id}}}
\newcommand\diag{\operatorname{diag}}

\newcommand\Dt{\frac{\mathrm{d}\phantom{t} }{\mathrm{d}\mspace{1mu}
t}}

\begin{document}
\mathtoolsset{
mathic 
}
\title{Top on a smooth plane}
\author{Maria Przybylska}
\email{m.przybylska@if.uz.zgora.pl}
\affiliation{Institute of Physics, University of Zielona G\'ora, \\
 Licealna 9, 65--417,  Zielona G\'ora, Poland}
\author{Andrzej J. Maciejewski}
\email{a.maciejewski@ia.uz.zgora.pl}
\affiliation{Janusz Gil Institute of Astronomy, University of Zielona G\'ora, \\
Licealna~9, 65-417,  Zielona G\'ora, Poland}

\begin{abstract}
We investigate the dynamics of a sliding top that is a rigid body with an ideal sharp
tip moving in a perfectly smooth horizontal plane, so no friction forces act
on the body. We prove that this system is integrable only in two cases
analogous to the Euler and Lagrange cases of the classical top problem. The
cases with the constant gravity field with acceleration $g\neq0$ and without
external field $g=0$ are considered. The non-integrability proof for $g\neq0$
based on the fact that the equations of motion for the sliding top are a
perturbation of the classical top equations of motion. We show that the
integrability of the classical top is a necessary condition for the
integrability of the sliding top.  Among four integrable classical top cases the
corresponding two cases for the sliding top are also integrable, and for the two
remaining cases, we prove their non-integrability by analyzing the differential
Galois group of variational equations along a certain particular solution. In
the absence of constant gravitational field $g=0$ the integrability is much more
difficult. At first, we proved that if the sliding top problem is integrable, then the
body is symmetric. In the proof, we applied one of the Ziglin theorem concerning
the splitting of separatrices phenomenon. Then we prove the non-integrability of
the symmetric sliding top using differential Galois group of variational
equations except two the same as for $g\neq0$ cases. The integrability of these
cases is also preserved when we add to equations of motion a gyrostatic term. 
\newline 
 \paragraph*{Declaration}The final published version of the article is available at \textbf{\href{https://doi.org/10.1063/5.0200592}{https://doi.org/10.1063/5.0200592} }
\end{abstract}

\maketitle
%
\begin{quotation}
A rigid body, apart from the material point, is the most important model in classical mechanics that represents various physical systems. The analysis of the dynamics of a rigid body with one fixed point in the gravitational field, called the heavy top, has been a research problem for hundreds of years and a testing ground for which various methods have been used and gave impetus to the creation of new methods. This paper provides a complete integrability analysis of an almost forgotten top model with an ideal sharp tip that slides on a perfectly smooth horizontal plane in a constant gravity field.
\end{quotation}

\section{Formulation of the problem and  main result}
\label{sec:equ}
We consider the constrained motion of a rigid body in a constant gravity field.
For the description of the problem, we use two orthonormal reference frames.  The
inertial frame $\mathscr{F}=\{O, \ve_1,\ve_2, \ve_3\}$ with origin at a point
$O$, and axes defined by three unit orthogonal vectors $\ve_i$ satisfying the 
relation $ \ve_3=\ve_1\times\ve_2$. The body-fixed frame
$\mathcal{F}_{\mathcal{B}}= \{C, \vb_1, \vb_2 , \vb_3\}$ has its origin in the
center of mass $C$ of the body,  and axes given by the unit orthogonal vectors
$\vb_i$ fulfilling $\vb_3=\vb_1\times\vb_2$. 

We will use the following convention. The coordinates of the vector $\vec{z}$ in
the inertial and body frames are denoted by $\vz=[z_1,z_2,z_3]^T$, and by
$\vZ=[Z_1,Z_2,Z_3]^T$, respectively. That is $z_i=\vz\cdot\ve_i$ and $Z_i=
\vz\cdot \vb_i$, for $i=1,2,3$. The orientation matrix $\vA=[A_{ij}]$ of the
body frame relative to the inertial frame is defined by
$A_{ij}=\ve_i\cdot\vb_{j}$  for  $i,j=1,2,3$, so   for an arbitrary vector
$\vec{z}$ we have $\vz=\vA\vZ$. Clearly, with our definitions, we have
$\vA^T\vA= \id$, and $\det\vA=+1$, so $\vA\in\mathrm{SO}(3,\R)$.  The time
derivative of $\vA(t)$ is given by
\begin{equation*}
  \dot \vA = \vA \widehat\vOmega,  
\end{equation*}
where $\widehat\vOmega = \vA^T \dot\vA=-\dot \vA^T\vA $  is an antisymmetric matrix. 

The standard isomorphism of Lie algebra $\R^{3}$ with cross product as the multiplication, and the Lie algebra $\mathfrak{so}(3,\R)$ of antisymmetric
matrices is given by 
\begin{equation}
  \R^3 \ni \vZ \longmapsto \widehat \vZ = \begin{bmatrix}
    0 & -Z_3 & Z_2\\
    Z_3 &0& -Z_1 \\
    -Z_2 & Z_1 & 0
  \end{bmatrix}\in\mathfrak{so}(3,\R).
  \label{eq:5}
\end{equation}
This isomorphism gives the vector $\vOmega$   corresponding to the matrix
$\widehat\vOmega$. It is the angular velocity of the body projected onto the
body axes. Moreover, if $\vz=\vA\vZ$, then $\widehat\vz =\vA\widehat\vZ\vA^T $,
so
\begin{equation}
  \label{eq:18}
  \dot\vA = \widehat\vomega \vA,
\end{equation}
where $\vomega=\vA\vOmega$ is the angular velocity in the inertial frame.

If a vector $\vec{z}$ is constant in the body frame, that is if
$\dot\vZ=\vzero$, then 
\begin{equation}
  \label{dotx}
  \dot\vz = \Dt \left(\vA\vZ\right)= \dot\vA \vZ =\widehat\vomega \vA\vZ= 
  \widehat\vomega\vz = \vomega\times \vz.  
\end{equation}
Moreover, we also have 
\begin{equation}
  \label{dotx2}
  \dot\vz = \vA \widehat\vOmega \vZ.
\end{equation}

The kinetic energy of the body $T$ is 
\begin{equation}
  \label{2T}
  T=
  \frac{m}{2} \dot\vr\cdot\dot\vr  + \frac{1}{2}
  \vOmega\cdot \vI\vOmega,
\end{equation}
where  $m$ is the body mass, $\vI$ is its tensor of inertia and $\vr$ is the radius vector of its center of mass. 
The potential energy of the body is given by 
\begin{equation}
  \label{U}
  U  = m g \ve_3 \cdot\vr,
\end{equation}
where $g$ denotes the gravitational acceleration. 

\begin{figure}[ht]
  \begin{center}
    \includegraphics*{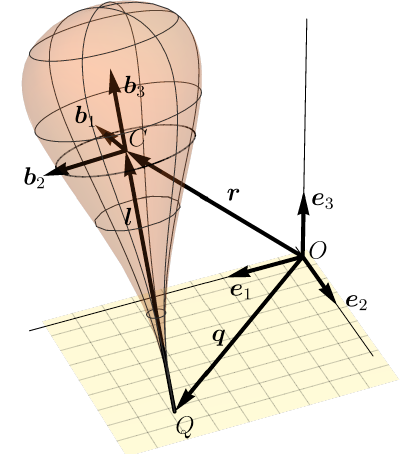}
  \end{center}
\caption{Geometry of the system.
  \label{geometry}}
\end{figure}

Our aim is to study the dynamics of a rigid body with an ideal sharp tip that is constrained to move in a horizontal plane which is assumed to be perfectly smooth, so no friction forces are acting on the body. We will call this model the sliding top.  The tip is a point $Q$ fixed in the body.  Let $\vq=\overrightarrow{OQ}$ be its radius vector, so we assume that $\vq=[q_1,q_2,0]^T$. We also denote $\vl = \vr-\vq$, see Fig.~\ref{geometry}. The point $Q$ is fixed
in the body, so $\vL = \vA^T\vl$ is a constant vector. 

The configuration of the system is given by a vector $\vq$ and the orientation
matrix $\vA$, so the configuration space  is $\R^2\times\mathrm{SO}(3,\R)$.
Because $\vr=\vq+\vl$ and $\vq\cdot\ve_3=0$, we can rewrite the kinetic and
potential energies in the form of
\begin{equation}
  T= \frac{m}{2}(\dot\vq +\dot\vl)\cdot (\dot\vq +\dot\vl)+ \frac{1}{2}
  \vOmega\cdot \vI\vOmega,\qquad U = m g \ve_3 \cdot\vl.
  \label{eq:kinipot}
\end{equation}
It is easy to notice that coordinates $q_1$ and $q_2$ are cyclic. The
corresponding momenta 
\begin{equation*}
  p_i = m( \dot q_i +\dot l_i), \qquad i=1,2,  
\end{equation*}
are first integrals of the system (the horizontal velocity of the system center of mass is constant). Without loss of generality, we fix their values to zero.  In
effect, the reduced Lagrange function of the system is
\begin{equation}
  \label{L} 
 \mathcal{L} = \frac{1}{2} m {\left(\ve_3\cdot\dot \vl\right)}^2 +\frac{1}{2} 
  \vOmega\cdot \vI\vOmega  - m g \ve_3 \cdot\vl.
\end{equation} 
Denoting $\vGamma =\vA^T\ve_3$, and using the fact that $\dot \vl = \vomega\times
\vl$, we can rewrite the Lagrange function \eqref{L} in the form 
\begin{equation}
  \label{Lf}
 \mathcal{L} = \frac{1}{2}\vOmega\cdot \vI\vOmega  +
  \frac{1}{2} m \left[\vGamma\cdot(\vOmega\times\vL)\right]^2 - 
  m g \vGamma\cdot\vL.
\end{equation}
Equations of motion  have the form 
\begin{equation}
  \label{eq:lag_borisov} 
  \begin{split}
    & \Dt \left(\pder{\mathcal{L}}{\vOmega}\right) = \pder{\mathcal{L}}{\vOmega}\times\vOmega +
    \pder{\mathcal{L}}{\vGamma}\times\vGamma, \\ 
    & \Dt  \vGamma = \vGamma \times \vOmega,
  \end{split}
\end{equation}
see,  \cite{Borisov:19::}.
Since 
\begin{equation}
  \label{eq:dL}
  \begin{split}
&\frac{\partial \mathcal{L}}{\partial \vOmega}=
\vI\vOmega+m[\vOmega\cdot(\vL\times\vGamma)](\vL\times\vGamma),\\
&\frac{\partial \mathcal{L}}{\partial \vGamma}=
m[\vGamma\cdot(\vOmega\times\vL)](\vOmega\times\vL)-mg\vL,
\end{split}
\end{equation}
the explicit form of equations \eqref{eq:lag_borisov} is
\begin{equation}
\begin{split}
&\left[\vI+m(\vL\times\vGamma)\otimes(\vL\times\vGamma) \right] 
\Dt\vOmega=\vI\vOmega\times \vOmega\\
&-
m\left[(\vGamma\times\vOmega)\cdot(\vOmega\times\vL)+
g\right](\vL\times\vGamma),\\
& \Dt  \vGamma = \vGamma \times \vOmega,
\end{split}
 \label{eq:lag_borisovexpl} 
\end{equation}
where $\otimes$ denotes the outer product of vectors, see e.g. \cite[Ch. 1]{Ortega:1987::}. 

Equations \eqref{eq:lag_borisovexpl}  have first integrals
\begin{equation}
\begin{split}
&H=  \frac{1}{2}\vOmega\cdot \vI\vOmega+ 
\frac{1}{2} m \left[\vGamma\cdot(\vOmega\times\vL)\right]^2 + 
m g \vGamma\cdot\vL,\\
&F_1= \vGamma^2,\quad F_2=\vGamma\cdot\vI\vOmega.
\end{split}
\label{integrals} 
\end{equation}

Let us  introduce the angular momenta defined as
\begin{equation}
\vM=\frac{\partial \mathcal{L}}{\partial \vOmega}  = 
\vJ \vOmega,\quad \vJ=\vI+m(\vL\times \vGamma)\otimes(\vL\times \vGamma),
\end{equation}
compare with formula (1.100) in  \cite{Borisov:19::}. 
Multiplying both sides of the first equation \eqref{eq:dL} by $\vGamma$ and by $\vL$,  we deduce that 
\begin{equation}
\begin{split}
&\vM\cdot\vGamma=  \vI\vOmega\cdot\vGamma=  \vI\vJ^{-1}\vM\cdot\vGamma,\\
&\vM\cdot\vL=  \vI\vOmega\cdot\vL=  \vI\vJ^{-1}\vM\cdot\vL. 
\end{split}
\end{equation}
Using these identities we find  the final form of equations of motion for the
sliding top
\begin{equation}
\begin{split}
 \Dt\vM&=\vM\times \vJ^{-1}\vM-mg\vL\times\vGamma\\
 &+ 
 m\left[\vGamma\cdot(\vJ^{-1}\vM\times\vL)\right]
 \left((\vJ^{-1}\vM\times\vL)\times\vGamma\right),\\
 \Dt  \vGamma &= \vGamma \times \vJ^{-1}\vM.
\end{split}
 \label{eq:lag_borisovMM} 
\end{equation}

\begin{figure}[ht]
  \begin{center}
    \includegraphics*{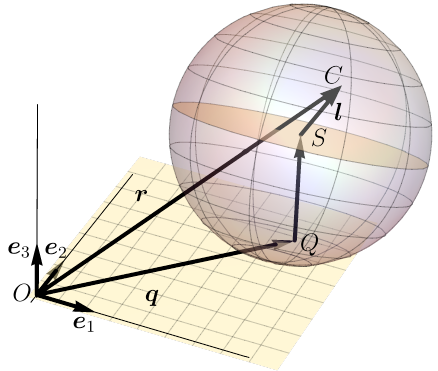}
  \end{center}
\caption{Another realization of the system.
  \label{geometry2}}
\end{figure}
Let us notice that the same equations of motion have another system: a ball with a displaced center of mass relative to its geometric center, sliding on a smooth
horizontal plane.  The difference is that unlike as is for a top, there are no
falling trajectories. The analogy between these two systems from the point of
view of permanent rotations was noted in \cite{Bizyaev:22::}. Indeed, if the
ball has radius $a$, then according to geometry presented in
Fig.~\ref{geometry2} the radius vector of the center of mass is
\[
\vr=\vq+a\ve_3+\vl,
\]
hence $\dot\vr=\dot\vq+\dot\vl$, and $U=mg\ve_3\cdot\vr=mg\ve_3\cdot\vl$, where we omitted the constant term $mga$. Thus, the kinetic $T$ and potential $U$ energies are the same as in \eqref{eq:kinipot}.

System~\eqref{eq:lag_borisovMM}  has first integrals
\begin{equation}
\begin{split}
&H=  \frac{1}{2} \vJ^{-1}\vM\cdot \vI \vJ^{-1}\vM+ 
\frac{1}{2} m \left[\vGamma\cdot( \vJ^{-1}\vM\times\vL)\right]^2 \\
&+ 
m g \vGamma\cdot\vL,\quad F_1= \vGamma^2,\quad F_2=\vGamma\cdot\vM.
\end{split}
\label{integralsM} 
\end{equation}
Equations of motion \eqref{eq:lag_borisovMM} can be written as
\begin{equation}
  \label{eq:poissonowska}
  \dot \vX = \J(\vX) \pder{H}{\vX}(\vX),\quad\text{where}\quad \vX= (\vM,\vGamma)
\end{equation}
and $\J(\vX)$ is $6\times6$ matrix with the following block structure
\begin{equation}
  \label{eq:Poisson}
  \J(\vX)= \begin{bmatrix}
    \widehat\vM &  \widehat\vGamma \\
    \widehat\vGamma & \vzero
  \end{bmatrix}.
\end{equation}
The matrix $\J(\vX)$ defines a degenerate Poisson bracket. For two
smooth functions $F(\vX)$ and $G(\vX)$ we define it as
\begin{equation}
  \label{eq:poissb}
  \{F,G\}(\vX) = \pder{F}{\vX}^T \J(\vX) \pder{G}{\vX}.
\end{equation}
Thus, the Poisson bracket of $\vM$ and $\vGamma$ given by $\J(\vX)$ is 
\[
\{M_i,M_j\}=-\varepsilon_{ijk}M_k,\,\,\,
 \{M_i,\Gamma_j\}=-\varepsilon_{ijk}\Gamma_k,\,\,\, \{\Gamma_i,\Gamma_j\}=0,
\]
so, it coincides with the Lie algebra $\mathfrak{e}(3)$. The functions $F_1$ and $F_2$ are Casimir functions of this Poisson structure. Thus, on a generic leaf,
the system is Hamiltonian and has two degrees of freedom.  For its
integrability, one additional first integral is missing, and it exists only for particular values of parameters of the system.  The system \eqref{eq:lag_borisovMM} depends on parameters: principal moments of inertia $A$, $B$ and $C$ that are eigenvalues of $\vI$ satisfying inequalities  
\[
\begin{split}
 &A>0,\qquad B>0,\qquad C>0,\\
 &A+B\geq C,\quad B+C\geq A,\quad C+A\geq B,
\end{split}    
\]
and components of vector $\vL$.

The main result of this paper is formulated in the following theorem.
\begin{theorem}
  \label{thm:we-main}
System~\eqref{eq:lag_borisovMM} is integrable with complex meromorphic first integrals only in the
following cases: 
\begin{enumerate}
\item The Euler case when $\vL=\vzero$ and the additional first integral is the
total angular momentum $F_3=\vM\cdot\vM$. The form of
equations~\eqref{eq:lag_borisovMM} shows that this case is not equivalent to the
case $g=0$ as in the classical heavy top. 
\item The Lagrange case where the body is symmetric and the center of mass lies on the axis of symmetry, for example, $B=A$ and $L_1=L_2=0$. Then the additional
first integral is the projection of the angular momentum onto the symmetry axis
$F_3=M_3$. In the completely symmetric case where $A=B=C$ the additional
first integral is $F_3=\vM\cdot\vL$.
\end{enumerate}
\end{theorem}
Moreover, we can also identify the Hess-Appelrot case. In this case, there
exists a polynomial of degree one which defines a five-dimensional invariant
hyperplane. As for the classical heavy top, without loss of generality, we
assume that $A\geq B\geq C$, and then we set  
\begin{equation}
L_2=0,\quad \sqrt{(A-B)C}\,L_3\pm\sqrt{(B-C)A}\,L_1=0. 
\label{eq:HAconditions}
\end{equation}
Under these conditions, the zero level of the polynomial    
\begin{equation}
  G =\sqrt{(A-B)C}\,M_1\mp \sqrt{(B-C)A}\,M_3 
 \label{eq:HAdarbus}
\end{equation}
is invariant.

All the above-mentioned cases were listed in \citep{Bizyaev:14::} in the context
of the study of a model with a ball with a displaced center of mass sliding in a smooth horizontal plane.

Notice that in Theorem~\ref{thm:we-main} we do not assume gravitational acceleration $g\neq 0$. The proof of this theorem splits naturally into two cases. In the first case, we assume that $g\neq 0$. Interestingly enough, the proof in the case $g=0$ is much more difficult.  

The plan for the rest of the article is as follows. In Section~\ref{sec:integrability_preliminar} we show the relation between the integrability of the considered system and the classical heavy top and illustrate its dynamics by means of the Poincar\'e cross-sections.  In Section~\ref{sec:Ziglin} analytical tools used for the integrability analysis are described. Sections~\ref{sec:anal} and~\ref{sec:g0} present proof of the main Theorem~\ref{thm:we-main} for cases in the presence of constant gravity field with acceleration $g\neq0$ and for $g=0$, respectively. In Section~\ref{sec:finalremarks} we summarize the results and show that the addition of a gyrostat term to two integrable cases of our system does not destroy its integrability. For interested readers, Appendices~\ref{sec:2ndorderlde} and \ref{sec:lamiaste} contain effective tools for checking differential Galois groups of second-order linear differential equations with rational coefficients and the special case of the Lam\'e equation necessary to follow applications of Theorem~\ref{thm:Morales}. For completeness, in Appendix~\ref{sec:residius_explicitely} are given some more complicated  expressions which appear in proof of Lemma~\ref{lem:3} during the application of Theorem~\ref{thm:Ziglin}.

\section{Integrability analysis. Preliminary considerations.}
\label{sec:integrability_preliminar}

\subsection{Relation of the system with classical heavy top}
\label{sec:withheavytop}

To compare  the considered system with the classical heavy top problem (a body with a fixed point in a constant gravity field), we  introduce the weights $\valpha=(1,1,1,2,2,2)$ for variables $\vX=(\vM,\vGamma)$, so that the weighted degree with respect to these weights is defined by   
\begin{equation*}
  \deg_{\valpha} M_i = 1, \qquad \deg_{\valpha} \Gamma_i = 2, \qquad i=1,2,3. 
\end{equation*}
Let $\vE_{\valpha}$ be the Euler field associated with these weights 
\begin{equation*}
  \vE_{\valpha} =\sum_{i=1}^6 \alpha_i X_i \pder{\phantom{X}}{X_i}.
\end{equation*}
Then, a function $F(\vX)$ is weight-homogeneous of degree $k$ if 
\begin{equation*}
  \vE_{\valpha}[F] = k F, 
\end{equation*}
and a  vector field $\vV(\vX)$ is weight-homogeneous of degree $k$ if 
\begin{equation*}
  [\vE_{\valpha}, \vV] = k\vV,
\end{equation*}
where here $[\cdot,\cdot]$ denotes the Lie bracket. 

The chosen weights are compatible with the Poisson structure defined by
$\J(\vX)$. Thus, if $F(\vX)$ is weight-homogeneous of degree $k$, then the
Hamiltonian vector field $\J(\vX) \tfrac{\partial F}{\partial\vX}$ is weight
homogeneous of degree $k-1$.  Moreover, if $F(\vX)$ and $G(\vX)$ are weight
homogeneous, then $\{F,G\}$ is weight-homogeneous and if $\{F,G\}\neq 0$ then
$\deg_{\valpha}\{F,G\}=\deg_{\valpha}F +\deg_{\valpha}G -1$.

Let $F(\vX)$ be a meromorphic function. Then it admits weight homogeneous
expansion 
\begin{equation*}
 F(\vX) =  F_m(\vX) + \cdots,
\end{equation*}
here $F_m(\vX)$ is a non-zero weight-homogeneous function of degree $m$, and
dots denote weight-homogeneous terms of degree higher than $m$. Function $F_m$
is called the lowest order term of $F$ and is denoted $F^\circ$. 
The Hamiltonian $H(\vX)$ is a rational function, so it has a unique expansion 
\begin{equation}
    H(\vX)  = H_2(\vX) + H_6(\vX) + \cdots,
    \label{eq:Hexpan}
\end{equation}
starting from the lowest second degree term $H_2(\vX)=H^\circ(\vX)$ and the next is of degree 6, which have explicit forms 
\begin{align}
  \label{eq:H2}
    H_2(\vX) = &\frac{1}{2}\vM\cdot  \vI^{-1}\vM + m g \vGamma\cdot\vL. \\ 
    \label{eq:H6}
    H_6(\vX) = &-\frac{1}{2} m \left[\vGamma\cdot( \vI^{-1}\vM\times\vL)\right]^2 .
\end{align}
Let us notice that Hamilton equations \eqref{eq:poissonowska} with $H=H_2$ and
the Poisson structure  given in \eqref{eq:Poisson} are the Euler-Poisson equations for the classical heavy top
\begin{equation}
\begin{split}
 &\Dt\vM=\vM\times \vI^{-1}\vM-mg\vL\times\vGamma,\\
 &\Dt  \vGamma = \vGamma \times \vI^{-1}\vM.
\end{split}
 \label{eq:heavytop} 
\end{equation}
To make this observation useful, we need the following lemma. 
\begin{lemma}
  \label{lem:lowest}
If the system~\eqref{eq:poissonowska} is integrable with meromorphic first
integrals, then system~\eqref{eq:heavytop} is integrable. 
\end{lemma}
\begin{proof}
By  the above-shown properties, the Poisson bracket of $F(\vX) =
F^\circ(\vX)+\cdots$ and $G(\vX) = G^\circ(\vX)+\cdots$  is 
\begin{equation*}
  \{F,G\} = \{F^\circ,G^\circ\} + \cdots 
\end{equation*}
Hence, if $F(\vX)$ is a first integral of~\eqref{eq:poissonowska} which is
functionally independent with $H$, $F_1$, and $F_2$, then
$F^\circ(\vX)$ is a first integral~\eqref{eq:H2}.
Furthermore, if $F$ commute with $F_1$ and $F_2$ then $F^\circ$ also commute with
$F_1$ and $F_2$. By the Ziglin Lemma \cite{Ziglin:82::b} we can always choose
$F$ in such a way that $H^\circ$, $F_1$, $F_2$ and $F^\circ$ are functionally
independent. 
\end{proof}
By the above lemma, the sliding top is integrable only in the cases for which
the heavy top is integrable. But the problem of the integrability of the heavy
top is solved completely.  In the following, we recall known facts.  

The heavy top equations of motion ~\eqref{eq:heavytop} have first integrals
$H_2$, $F_1$ and $F_2$ given by \eqref{integralsM}. For their integrability,
just one additional first integral is necessary. The problem of the
integrability of these equations has been analyzed for centuries and the following integrable cases have been identified.
\begin{enumerate}
\item The Euler case (1758) corresponds to the situation when there is
  no gravity (i.e. when $g=0$) or $\vL=\vzero$ (the fixed point of
  the body is the center of mass). The additional first integral in
  this case is the total angular momentum $F_3 = \vM\cdot\vM$.
\item In the Lagrange case~\cite{Lagrange:1889::} the body is symmetric (i.e. two of its principal moments of inertial are equal)
  and the fixed point lies on the symmetry axis. The additional first integral in this case is the projection of the angular momentum onto the symmetry axis. If we assume that $A=B$, then in the Lagrange case $L_1=L_2=0$, and $F_3=M_3$.
\item In the Kovalevskaya case~\citep{Kowalevski:1888::,Kowalevski:1890::} the
body is symmetric and the principal moment of inertia along the symmetry axis is
half of the principal moment of inertia with respect to an axis perpendicular to
the symmetry axis. Moreover, the fixed point lies in the principal plane
perpendicular to the symmetry axis. If $A=B=2C$, then (after an appropriate
rotation around the symmetry axis) we have in the Kovalevskaya case $L_2=L_3=0$.
The additional first integral has the form
 \begin{multline*}
F_3 = \left( \frac{1}{2}(M_1^2-M_2^2)-mg L_1 A\Gamma_1\right)^2 \\
+(M_1M_2-mg L_1 A \Gamma_2)^2.
\end{multline*}
\item In the Goryachev-Chaplygin case~\citep{ Goryachev:1910::} the body is
symmetric and, as in the Kovalevskaya case, the fixed point lies in the
principal plane perpendicular to the symmetry axis. If we assume that the third
principal axis is the symmetry axis, then in this case we have $A=B=4C$ and
$L_2=L_3=0$. In the Goryachev-Chaplygin case equations~\eqref{eq:heavytop}  are
integrable only at the level $F_1=0$ and the additional first integral has the
following form:
 \[
F_3 =
  M_3(M_1^2+M_2^2) -mg L_1 A M_1\Gamma_3.
\]
\end{enumerate}

For more details on recent studies and results in rigid body dynamics, mathematical structures, interpretations, and generalizations of these cases, see \cite{Borisov:19::}.
Furthermore, for many years, there was an open question of whether the list of integrable cases above is complete. The first important step in answering this
question was taken by \cite{Ziglin:80::d}, who proved the following theorem. 
\begin{theorem}[Ziglin, 1980]
\label{thm:z1}
If $(A-B)(B-C)(C-A)\neq 0$ and $mg\vL\neq\vzero$, then the Euler-Poisson equations \eqref{eq:heavytop} 
does not admit a real meromorphic first integral which is functionally independent together with $H_2$, $F_1$, and $F_2$.
\end{theorem}
The proof of this theorem is based on the original method of splitting the separatrices developed by Ziglin, see also for details \cite{Ziglin:80::e}. 
More precisely, the proof consists in the application of Theorem~\ref{thm:Ziglin} given in Section~\ref{sec:Ziglin}.  
The final answer to the integrability problem was also given by Ziglin in~\cite{Ziglin:83::b}. Using his elegant theory based on properties of the monodromy group of variational equations, he proved the following theorem.  
\begin{theorem}[Ziglin, 1983]
\label{thm:z2}
The complexified Euler-Poisson system for a symmetric body is integrable on the
level $F_2=0$ with complex meromorphic first integrals only in the four classical cases.
\end{theorem}
The non-integrability with real meromorphic first integrals was proved
in~\cite{Ziglin:97::}. Similar results for a symmetric heavy top were given by the authors of the differential Galois group of variational equations in~\cite{mp:05::a}.

\subsection{Numerical analysis of the system}
\label{sec:numerical}

In the previous section we showed how the sliding top is related to the classical heavy top.  But immediately we notice the differences: namely in our Theorem~\ref{thm:we-main}  the Kovalevskaya case does not appear, moreover, we claim that even when the gravity field vanishes the system is not integrable.
Simple numerical experiments show that both systems are fundamentally different. 

We numerically integrate the equations~\eqref{eq:lag_borisovMM}. For chosen values
of parameter $\vI$, $\vL$, $m$ and $g$ we fix a common level of the first integrals
$H$, $F_1$ and $F_2$. Generically, it is three-dimensional. On this level, we choose a cross-section plane. We mark on it the points where an orbit passes through this plane in a specified direction. For the presentation of the results, we use the
Androyer-Deprit variables $(\aG,\aL,\aH,\ag,\al, \ah)$, see \cite{Borisov:19::}.
As the cross-section plane, we choose $\ag=\pi/2$ and the cross direction is fixed by
$\dot\ag>0$. As coordinates on this plane, we choose $(\al,\aL/\aG)$.

\begin{figure*}[htp]
   \subfigure[ Classical top\label{fig:topasym} ]{
      \includegraphics[width=0.48\textwidth]{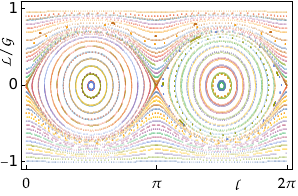}}\hspace{-.1cm}
    \subfigure[Sliding top \label{fig:sliderasym}]{
  \includegraphics[width=0.48\textwidth]{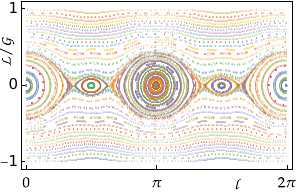}}
\caption{Poincar\'e cross-sections for asymmetric classical~(a),
and sliding~(b) tops when $g=0$. For the sliding top $\vL=\left(\tfrac{8}{9},0,0\right)$.
\label{fig:asymmetric}}
\end{figure*}

\begin{figure*}[htp]
   \subfigure[$L=|\vL|=1/4$\label{fig:sliderasymdL1n4}]{
      \includegraphics[width=0.48\textwidth]{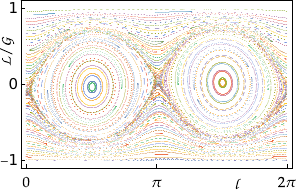}}\hspace{-.1cm}
    \subfigure[$L=|\vL|=1/2$ \label{fig:sliderasymdL1n2}]{
  \includegraphics[width=0.48\textwidth]{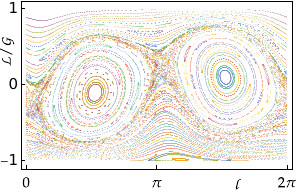}}
\caption{Poincar\'e cross-sections for asymmetric  sliding~top when $g=0$ for
two different lengths of $L=|\vL|$: (a) 1/4 and (b) 1/2.  Values of $\vL$: (a) 
$\vL=(1/6,1/8,\sqrt{11}/24)$, (b) $\vL=(1/3,1/4,\sqrt{11}/12)$.
\label{fig:asymmetricdifferentL}}
\end{figure*}

In all the examples below, we assume that both tops have the same moments of inertia
$I_1=1$, $I_2=\tfrac{2}{3}$ $I_3=\tfrac{1}{2}$ and mass $m=1$. The fixed values
of first integrals are $H=\tfrac{3}{2}$ and $F_2=\tfrac{1}{20}$. 

At first,  we consider the case without gravity ($g=0$). The dynamics in the case of the classical top shows in Fig.~\ref{fig:topasym} where we notice two unstable
periodic solutions close to points $(\al,\aL)=(0,0)$ and $(\al,\aL)=(\pi,0)$ in
the cross-section plane. They correspond to the unstable rotation of the top around
the second principal axis. Stable periodic solutions corresponding to rotations
around the first principal axis are visible at points $(\al,\aL)=(\pi/2,0)$ and
$(\al,\aL)=(3\pi/2,0)$. They are surrounded by
quasi-periodic solutions.  Fig.~\ref{fig:sliderasym} shows the dynamics of the sliding top with $\vL=\left(\frac{8}{9},0,0\right)$. We notice that it is completely different from the one shown in Fig.~\ref{fig:topasym}. Now near the
points $(\al,\aL)=(0,0)$ and $(\al,\aL)=(\pi,0)$ stable periodic solutions appear and additionally also near points $(\al,\aL)=(\pi/2,0)$ and
$(\al,\aL)=(3\pi/2,0)$. Moreover, four unstable periodic solutions appear and in their vicinity we notice a chaotic behavior. 

\begin{figure*}[htp]
   \subfigure[Classical top, $g=1$ and $L_1=L=1$\label{fig:rigidKo}]{
      \includegraphics[width=0.48\textwidth]{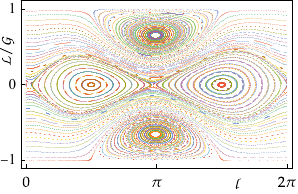}}\hspace{-.1cm}
    \subfigure[Sliding top, $g=20$ and $L_1=L=\tfrac{1}{20}$ \label{fig:slidedKo}]{
  \includegraphics[width=0.48\textwidth]{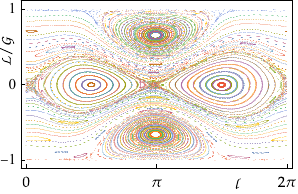}}
\caption{Poincar\'e cross-sections for the Kovalewskaya cases. Values of
parameters are $I_1=I_2=1$, $I_3=\tfrac{1}{2}$, $L_2=L_3=0$  and $m=1$ and
values of first integrals $H=\tfrac{3}{2}$ and $F_2=\tfrac{1}{20}$.
 \label{fig:Kowal}}
\end{figure*}

However, in the case of no gravity, the sliding top can be considered as a
perturbation of the classical free top. The role of a small parameter plays the
length of the vector $\vL$. Even for relatively big values of $L$ the differences
between cross sections for the classical and sliding tops are small. This is why we show in Fig.~\ref{fig:asymmetricdifferentL} results for two cases with
$L=1/4$ and $L=1/2$.

\begin{figure*}[htp]
   \subfigure[Classical top, $g=1$ and $L_1=L=1$ \label{fig:rigidGC}]{
      \includegraphics[width=0.48\textwidth]{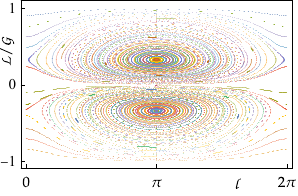}}\hspace{-.1cm}
    \subfigure[ Sliding top, $g=20$ and $L_1=L=\tfrac{1}{20}$ \label{fig:slidedGC}]{
  \includegraphics[width=0.48\textwidth]{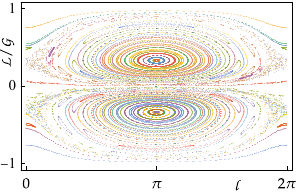}}
\caption{Poincar\'e cross-sections for
Goryachev-Chaplygin cases. The parameters values are $I_1=I_2=1$,
$I_3=\tfrac{1}{4}$, $L_2=L_3=0$  and $m=1$ and the values of the first integrals
$H=\tfrac{3}{2}$ and $F_2=0$. 
\label{fig:GCY}}
\end{figure*}

\begin{figure*}[htp]
   \subfigure[Classical
top, $g=1$ and  $L_1=\tfrac{3}{5}$, $L_3=\tfrac{4}{5}$ \label{fig:rigidHA}]{
      \includegraphics[width=0.48\textwidth]{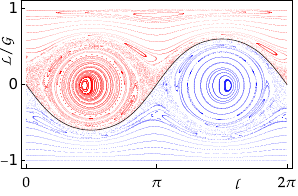}}\hspace{-.1cm}
    \subfigure[Sliding top,  $g=2$ and $L_1=\tfrac{3}{10}$, $L_3=\tfrac{4}{10}$
\label{fig:slidermult2}]{
  \includegraphics[width=0.48\textwidth]{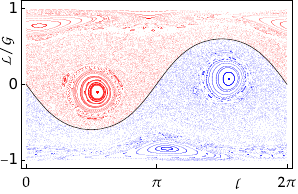}}
\caption{Poincar\'e cross-sections for the Hess-Appelrot cases. The parameters values are $I_1=1$, $I_2=\tfrac{5}{8}$, $I_3=\tfrac{3}{8}$, $L_2=0$, $m=1$
and the values of the first integrals $H=28$ and $F_2=3$. 
The cross-points obtained from initial conditions that satisfy $G>0$ and $G<0$ are
plotted in red and blue, respectively. \label{fig:HAY}}
\end{figure*}

The expansion~\eqref{eq:Hexpan} shows that we can consider the slider top
as a perturbation of the classical heavy ($g\neq0)$ top. From the forms
\eqref{eq:H2} and~\eqref{eq:H6}, we deduce that it is enough to fix $gL=1$, and
then consider $L$ as a perturbing parameter. 

Figures~\ref{fig:Kowal} and~\ref{fig:GCY} show Poincar\'e cross sections for the
Kovalevskaya and Goryachev-Chaplygin cases for the classical top (on the left) and the corresponding sliding top (on the right).

The cross sections presented in Fig.~\ref{fig:Kowal} were obtained by numerical integration equations of motion for values of parameters $I_1=I_2=1$,
$I_3=\tfrac{1}{2}$, $L_2=L_3=0$  and $m=1$. We chose for the classical top $g=1$ and
$L_1=L=1$ and for the sliding top $g=20$ and $L_1=L=\tfrac{1}{20}$ to have the
same values of $mgL=1$ for both tops. We fix the same values of the first integrals $H=\tfrac{3}{2}$ and $F_2=\tfrac{1}{20}$.  For both systems, one
can notice two unstable periodic solutions near $(\al,\aL)=(0,0)$.
$(\al,\aL)=(\pi,0)$. There are also two pairs of stable periodic solutions
surrounded by quasi-periodic trajectories. The left Fig.~\ref{fig:rigidKo} shows an integrable Kovalevskaya rigid top.  Fig.~\ref{fig:slidedKo} showing the cross section for the sliding top regions shows chaotic behavior appearing near unstable periodic orbits. The origin of this behavior is connected with the separatrices crossing phenomenon.

In Fig~\ref{fig:GCY} equations of motion were integrated for values of parameters $I_1=I_2=1$, $I_3=\tfrac{1}{4}$, $L_2=L_3=0$  and $m=1$. Similarly, as previously, we chose the rigid top $g=1$ and $L_1=L=1$, and for the sliding top $g=20$ and
$L_1=L=\tfrac{1}{20}$. We consider the same values for the first integrals
$H=\tfrac{3}{2}$ and $F_2=0$. On the left, Fig.~\ref{fig:rigidGC} corresponding
to the integrable heavy top traces of three periodic solutions are visible: two
hyperbolic with $\al=0$ and elliptic with $\al=\pi$ surrounding by
quasi-periodic orbits.  On the right Fig.~\ref{fig:slidedGC} corresponding to
the sliding top the chaotic layers around separatrices to two hyperbolic
equilibria are visible.

Figures~\ref{fig:HAY} illustrate the Hess-Appelrot cases of the classical heavy
top Fig.~\ref{fig:rigidHA} and the corresponding sliding top
Fig.~\ref{fig:slidermult2}. We choose the following values of parameters
that satisfy conditions \eqref{eq:HAconditions} $I_1=1$, $I_2=\tfrac{5}{8}$,
$I_3=\tfrac{3}{8}$, $L_2=0$, $m=1$ and for the heavy top $g=1$ and
$L_1=\tfrac{3}{5}$, $L_3=\tfrac{4}{5}$ and for the sliding top $g=2$ and
$L_1=\tfrac{3}{10}$, $L_3=\tfrac{4}{10}$. In both figures are visible wide
layers of randomly distributed crossing points separated by the zero level of
Darboux polynomial $G=\sqrt{\frac{3}{5}}M_1+\frac{4}{\sqrt{15}}M_3$ which in
Androyer-Deprit variables takes the form $G=4\aL+3 \sqrt{\aG^2-\aL^2} \sin\al$.
This curve separates cross-points obtained from initial conditions satisfying
$G>0$ and $G<0$ that are plotted in red and blue, respectively.

\section{Non-integrability theorems}
\label{sec:Ziglin}

The analysis of the integrability of systems of high dimensions and dependent on
parameters is complicated. One can expect generically non-integrability, but for
certain values of parameters the system can be integrable. To select such
values, one needs theorems formulating necessary conditions for integrability,
or just necessary conditions for the existence of a certain number of additional
functionally independent first integrals. 

In this paper, we will use two such theorems. One of them is related to the
phenomenon of splitting of surfaces asymptotic to a hyperbolic periodic
solution of a system which is a perturbation of an integrable system; for more details, see, e.g. \cite[Sec. 7.2]{Arnold:06::} or \cite[Ch. V]{Kozlov:96::}.  

Here, we use the formulation proposed
in \cite{Ziglin:80::d,Ziglin:80::e}. 
Let us consider a Hamiltonian system given by Hamiltonian 
\begin{equation}
    H(q,p, \varphi, I, \mu) =  H_0(q,p,  I) +\mu H_1(q,p, \varphi, I)+\cdots,
    \label{eq:HexpZiglin}
\end{equation}
which is a real analytic function $2\pi$-periodic with respect to $\varphi$; variables $\vx=(q,p)$ belong to an open set $U\subset\R^2$, $I\in (I_0-\epsilon,I_0+\epsilon)$, 
$\varphi\in\R \mod 2\pi$, and $\abs{\mu}<\epsilon$. 

Hamilton's equations take the form
\begin{equation}
\frac{\mathrm{d} q}{\mathrm{d} t}=\frac{\partial H}{\partial p},\,\,\,
\frac{\mathrm{d} p}{\mathrm{d} t}=-\frac{\partial H}{\partial q},\,\,\,
\frac{\mathrm{d} \varphi}{\mathrm{d} t}=\frac{\partial H}{\partial I},\,\,\,
\frac{\mathrm{d} I}{\mathrm{d} t}=-\frac{\partial H}{\partial \varphi},
\label{eq:sysZig}
\end{equation}
and for $\mu=0$ unperturbed Hamiltonian system is
\begin{equation}
\begin{split}
&\frac{\mathrm{d} q}{\mathrm{d} t}=\frac{\partial H_0(q,p,  I)}{\partial p},\quad
\frac{\mathrm{d} p}{\mathrm{d} t}=-\frac{\partial H_0(q,p,  I)}{\partial q},\\
&
\frac{\mathrm{d} \varphi}{\mathrm{d} t}=\frac{\partial H_0(q,p,  I)}{\partial I},\quad
\frac{\mathrm{d} I}{\mathrm{d} t}=0.
\end{split}
\label{eq:sysZigunpertur}
\end{equation}
We assume that:\\
1. the  system
\begin{equation}
\frac{\mathrm{d} q}{\mathrm{d} t}=\frac{\partial H_0(q,p,  I_0)}{\partial p},\quad
\frac{\mathrm{d} p}{\mathrm{d} t}=-\frac{\partial H_0(q,p,  I_0)}{\partial q},
\label{eq:nonperturb}
\end{equation}
has two  hyperbolic equilibria $\vx_1$ and $\vx_2$ (not necessarily distinct),
and the eigenvalues $\lambda_i$ and $-\lambda_i$ of the linearization of the
system~\eqref{eq:nonperturb} at the point $\vx_i$, $i=1,2$ are real and
different from zero;\\
2. equilibria $\vx_1$ and $\vx_2$ are joined 
by double-asymptotic solution $\widehat{\vx}(t)$ such that 
\[
\lim_{t\to-\infty}\widehat{\vx}(t)=\vx_1,\quad\text{and}\quad \lim_{t\to\infty}\widehat{\vx}(t)=\vx_2;
\]
3.  Inequality $\tfrac{\partial H_0}{\partial I}(\widehat{\vx}(t),I_0)>c>0$ holds for all $-\infty<t<\infty$.

Let us introduce the notation
\[
\begin{split}
 &\widehat{\vz}(t)= (\widehat{\vx}(t),I_0,\widehat{\varphi}(t)),\quad
 \widehat{\varphi}(t)=\int_0^t\frac{\partial H_0}{\partial I}(\widehat{\vx}(\tau),I_0)d\tau,\\
 & \widehat{\vz}(t,\overline{\varphi})=(\widehat{\vx}(t),I_0,\widehat{\varphi}(t)+\overline{\varphi}),
\end{split}
\]
and make the following further assumptions.\\
4. Solution $\widehat{\vz}(t)$ is analytical, although generally it has an
analytic continuation (in general not single-valued) with at most a finite
number of singular points in the strip $\Pi:0\leq \mathrm{Im}\,
t\leq\tfrac{2\pi}{\lambda_1}$. We denote by $\Pi'$ the strip $\Pi$ with these
singular points removed;\\
5. Hamiltonian $H(q,p, \varphi, I, \mu)$ is analytic, although generally it has
an analytic continuation (in general not single-valued) in a domain of complex
$(q,p,\varphi,I,\mu)$-space containing the point
$(\widehat{\vz}(t,\overline{\varphi}),\mu=0)$ for every $t\in\Pi'$ and
$\overline{\varphi}\in\R$;\\
6. The functions $\frac{\partial H_0}{\partial I}(\widehat{\vz}(t))$ and
$\frac{\partial H_1}{\partial \varphi}(\widehat{\vz}(t,\overline{\varphi}))$ for
fixed $\overline{\varphi}\in\R$ are single-valued in $\Pi$.

In \cite{Ziglin:80::d}, see also \cite{Ziglin:80::e} the following result was proved.
\begin{theorem}
\label{thm:Ziglin0}
If for at least one value of  $\overline{\varphi}\in\R$ the function
$\frac{\partial H_1}{\partial \varphi}(\widehat{\vz}(t,\overline{\varphi}))$ has
nonzero residues in $\Pi$, then there exists a solution of \eqref{eq:sysZig}
which is not single-valued in the sense that
\[
\lim_{\mu\to0}\frac{\Delta I(\mu)}{\mu}\neq0,
\]
where $\Delta I(\mu)$ is the increment of the coordinate $I(t,\mu)$ of
$(q,p,\varphi,I)(t,\mu)$ after one circuit of closed contour
$\Gamma\subset\Pi'$. If the sum of these residues is not equal to zero, then for
any $|\mu|\neq0$ small enough, the system  \eqref{eq:sysZig} does not have a
first integral analytic in $U$ and functionally independent of $H$.
\end{theorem}

In this paper we will use a more convenient formulation of these integrability
obstructions given in \cite{Ziglin:80::e}. We expand $H_1$ in the Fourier series
\[
\begin{split}
&H_1(q,p, \varphi, I)=\sum_{k\in\Z} h_k(q,p, \varphi, I),\\
&h_k(q,p, \varphi, I)=\overline{h}_k(q,p,  I)e^{\rmi k\varphi},
\end{split}
\]
and then we have the following theorem. 
\begin{theorem}
\label{thm:Ziglin}
If for at least one $k\neq0$  the sum of residues of function $h_k(q,p, \varphi,
I)$  is non-zero, then not even one small enough $|\mu|$ has in phase space $U$
analytic first integral independent of $H$.
\end{theorem}
This theorem applies to systems that can be considered as Hamiltonian
systems that are perturbations of appropriate integrable systems with two
immovable hyperbolic equilibria.
For the proof and details, see \cite{Ziglin:80::d,Ziglin:80::e}. We will use this
theorem in the proof of Theorem~\ref{thm:g0gen} which is, in a certain sense, a
version of Ziglin Theorem~\ref{thm:z1} for the sliding top without gravity.
Amazingly enough, the proof of the Ziglin Theorem~\ref{thm:z1} is quite simple, but the
proof of our theorem~\ref{thm:g0gen} is quite complicated and laborious.

In the case when for a considered non-linear system a non-equilibrium particular solution can be explicitly written, variational equations along this solution
can be calculated. The presence of first integrals and the integrability of a
nonlinear system implies the presence of the same number of first integrals and
the integrability of the variational equations. 

Moreover, variational equations are linear equations and for them the monodromy
group and the differential Galois group are well defined. The monodromy group
acts linearly in a solution space of a linear system (or scalar linear equation
of a certain degree) by analytic continuations of solutions along closed loops.
More precisely, the monodromy matrices form an anti-representation of the first
homotopy group of a Riemann surface related to the particular solution. First
integrals of variational equations appear invariant of the monodromy group. The
integrability obstructions can be translated to the properties of the monodromy
group of variational equations. They were formulated by \cite{Ziglin:82::b}.
This approach is now called
the Ziglin theory.

Since the monodromy group is just a matrix group, more useful for applications is
the differential Galois group which is an algebraic group. The aim of the
differential Galois theory is to study the question of the solvability of linear
differential equations with coefficients in a certain differential field.
Usually, solutions do not belong to the basic differential field containing
coefficients of the considered equation.  The smallest differential field
containing all solutions is called the Picard-Vessiot extension of the basic
field. The differential Galois group is the group of automorphisms of the
Picard-Vessiot extension (i.e. invertible transformations of preserving field
operations) that commute with differentiation and do not change elements of the
basic differential field.  For a detailed exposition of the differential Galois
theory, see, e.g. \cite{Kaplansky:76::,Morales:99::c} or a short introduction
focused on applications in \cite{mp:03::b}.

Differential Galois group is an algebraic group that has a few components, and
this one containing the identity is called the identity component. The first
integrals of variational equations are also invariant of the differential Galois
group. The conditions of the integrability in the Liouville sense for Hamilton
equations translate in the properties of the differential Galois group
formulated in the following theorem due to J.J Morales-Ramis and J.-P. Ramis,
for details see \cite{Morales:99::c}.
\begin{theorem}
If a Hamiltonian system is meromorphically integrable in the Liouville sense in
a neighborhood of a phase curve $\Gamma$ corresponding to a particular
solution, then the identity component of the differential Galois group of
variational equations along $\Gamma$ is Abelian.
\label{thm:Morales}
\end{theorem}
Despite its abstract definition, the differential Galois group is known for
various linear equations, e.g. the hypergeometric equation or the Lam\'e
equation, and in the case of second-order linear equations with rational
coefficients, there is an algorithm, called the Kovacic algorithm, see
Appendix~\ref{sec:2ndorderlde}, that always allows its determination. We will
use this theorem for the analysis of cases where the considered sliding top is
symmetric.

\section{Integrability analysis. Case $g\neq 0$}
\label{sec:anal}

In this section, we assume that $g\neq 0$ and prove the following theorem.
\begin{theorem}
  \label{thm:gneq0}
If $g\neq 0$ then the sliding top problem is not integrable in the class  of complex meromorphic first integrals, except the Euler and
the Lagrange cases.
\end{theorem}
\begin{proof}
According to Lemma~\ref{lem:lowest} the necessary integrability condition for
our system is the integrability of the corresponding classical heavy top. The
list of known integrable case for classical heavy is recalled in
Section~\ref{sec:withheavytop}. We already noticed that the sliding top is
integrable in the Euler case with $\vL=0$, and in the Lagrange case. Thus,
we should analyze only the Kovalevskaya and the Goryachev-Chaplygin cases. In
both of these cases we can assume that ${L}_2={L}_3=0$ and $B=A$.

One can easily notice that manifold 
\begin{equation*}
 \cN:= \{ ( \vM,\vGamma)\in \C^6\ | \ M_1=M_3=\Gamma_2 = 0,
 \ \Gamma_1^2+\Gamma_3^2=1\},   
 \label{eq:inva1}
\end{equation*}
is invariant with respect to flow of \eqref{eq:lag_borisovexpl}.
 The system restricted to $ \cN$ takes the form
\begin{equation}
  \label{eq:rep}
  \begin{split}
 & \Dt M_2 = \Gamma_3+\frac{M_2^2\Gamma_1\Gamma_3}{(A+\Gamma_3^2)^2}, \quad 
  \Dt \Gamma_1=-\frac{M_2\Gamma_3}{A+\Gamma_3^2}, \\
 & \Dt \Gamma_3=\frac{M_2\Gamma_1}{A+\Gamma_3^2},
  \end{split}
\end{equation}
and it gives a family of particular solutions obtained as the intersection of
two algebraic curves
\begin{equation}
h=\frac{M_2^2}{2 (b+\Gamma_3^2)}+\Gamma_1,\quad \Gamma_1^2+\Gamma_3^2=1.  
\label{eq:part}
\end{equation}

If we define by $(\vm,\vgamma)$ variations of variables  $(\vM,\vGamma)$
then the variational equations take the form
\begin{equation*}
  \label{eq:varallv}
  \Dt \begin{bmatrix}
    m_1\\
    m_2\\
    m_3 \\
    \gamma_1 \\
    \gamma_2\\
    \gamma_3
  \end{bmatrix}
  =
\begin{bmatrix}
  0 & 0 & a_{13} & 0 & a_{15} & 0 \\
  0 & a_{22} & 0 & a_{24} & 0 & a_{26}\\
  a_{31} & 0 & a_{33} & 0 & a_{35} & 0 \\
  0 & a_{42} & 0 & 0 & 0 & a_{46}\\
  a_{51} & 0 & a_{53} & 0 & a_{55} & 0 \\  
  0 & a_{62} & 0 & a_{64} & 0 & a_{66}\\
\end{bmatrix}
\begin{bmatrix}
  m_1\\
  m_2\\
  m_3 \\
  \gamma_1 \\
  \gamma_2\\
  \gamma_3
\end{bmatrix}.
\end{equation*}
Here coefficients matrix $[a_{ij}]$ is just matrix of derivative of the
right-hand sides of equations of motion~\eqref{eq:poissonowska} evaluated at the
particular solution. 

Notice that equations for variables $(m_1, m_3, \gamma_2)$ form a close
subsystem which in explicit form reads
\begin{equation}
  \label{eq:normal_sys}
 \hspace{-0.3cm} \begin{bmatrix}
    \dot m_1\\[2.3ex]  \dot m_3 \\[2.3ex]  \dot\gamma_2 
   \end{bmatrix}\hspace{-0.15cm}
  =  \hspace{-0.15cm} \begin{bmatrix}
0 &\frac{(A-C)M_2}{C \left(A+\Gamma_3^2\right)}&\frac{(A-C)\Gamma_3 M_2^2}{C \left(A+\Gamma_3^2\right)^2}\\[0.5em]
  -\frac{M_2 \Gamma_3^2}{A \left(A+\Gamma_3^2\right)}&\frac{\Gamma_1 \Gamma_3 M_2}{C(A +\Gamma_3^2)}&\frac{\Gamma_1 \Gamma_3^2 M_2^2}{C \left(A+\Gamma_3^2\right)^2}-1\\[0.5em]
    \frac{\Gamma_3}{A}&-\frac{\Gamma_1}{C}&-\frac{\Gamma_1\Gamma_3 M_2}{ C(A+\Gamma_3^2)}
  \end{bmatrix}
  \hspace{-0.15cm}
   \begin{bmatrix}
     m_1\\[2.3ex]   m_3 \\[2.3ex] \gamma_2
   \end{bmatrix}.
\end{equation}
It is the system of normal variational equations that describes variations from normal to
invariant manifold $\cN$. 
It has the first integral 
\[
 f_2 = \Gamma_1 m_1   +  \Gamma_3m_3 + 
 M_2 \gamma_2.
\]
We consider the zero level of this integral and express $\gamma_2$ as a function
of $m_1$ and $m_3$. In effect, we obtain a system of two equations from which
we obtain one second-order equation for $m_1$
\begin{equation}
  \label{eq:rnvegneq0}
   \Dtt m + a_1(t) \Dt m + a_0(t) m = 0, \qquad m \equiv m_1,
\end{equation}
with coefficients
\[
\begin{split}
& a_1(t) =-\frac{2\Gamma_3}{M_2} +\frac{2 M_2 \Gamma_1 \Gamma_3}{(A + \Gamma_3^2)^2},\\
&a_0(t)=-\frac{\Gamma_1 (A-C) \left[\left(A+\Gamma_3^2\right)^2-\Gamma_1 M_2^2\right]}{C
   \left(A+\Gamma_3^2\right)^3}.
\end{split}
\]
Now, we make the following transformation of the independent variable
\begin{equation}
\label{eq:tran}
t \longrightarrow z := \frac{\Gamma_3(t)}{1+\Gamma_1(t)}.
\end{equation}
Then we obtain
\[
\begin{split}
&\Gamma_1 = \frac{1-z^2}{1+z^2}, \qquad \Gamma_3 = \frac{2 z}{1+z^2},   \\
 &M_2 = \frac{2\dot z(A + 4 z^2 + 2 A z^2 + A z^4)}{(1+z^2)^3}.
\end{split}
\]
Using the first equations in \eqref{eq:part} and \eqref{eq:rep} one can
calculate derivatives
\[
\begin{split}
&\dot z^2=\frac{(z^2+1)^3 \left[(h+1) z^2+h-1\right]}{2 \left[A (z^2+1)^2+4 z^2\right]},\\
&\ddot z=\frac{z (z^2+1)^2}{\left[A (z^2+1)^2+4 z^2\right]^2} \Big[h (z^2+1) \left(A (z^2+1)^2+6 z^2-2\right)\\
&+A
   \left(z^3+z\right)^2+6 z^4-4 z^2+2\Big]
\end{split}
\]
and transform equation \eqref{eq:rnvegneq0} into a linear equation with rational coefficients of the form
\begin{equation}
\label{eq:var}
m''+p(z)m'+q(z)m=0,\qquad '=\dfrac{\mathrm{d}}{\mathrm{d}z},
\end{equation}
with coefficients  
\begin{equation}
\begin{split}
 &p(z)= \frac{\ddot z+a_1\dot z}{\dot z^2} =
z \Big[\frac{2 \left(A z^2+A+2\right)}{A (z^2+1)^2+4 z^2}+\frac{1}{z^2+1}\\
&-\frac{h+1}{(h+1)
   z^2+h-1}\Big],\\
 &q(z)=\frac{a_0}{\dot z^2}=\frac{A-C}{C}\Big[-\frac{2}{(z^2+1)^2}+\frac{h}{z^2+1}\\
 &-\frac{h (h+1)}{h
   z^2+h+z^2-1}+\frac{4 (A+1)}{A (z^2+1)^2+4 z^2}\Big].
\end{split}
\label{eq:ppqqdgneq0}
\end{equation}
We fix $h$ such that $1 + A - h^2=0$.  Then polynomial $A (z^2+1)^2+4 z^2$ divides by $(h+1)z^2+h-1$, two singularities of this equation vanishes and coefficients $p(z)$ and $q(z)$ in \eqref{eq:ppqqdgneq0} simplify to 
\begin{equation}
\begin{split}
&p(z)=   \frac{2 z \left((h-1) z^2+h\right)}{h (z^2+1)^2-z^4+1},\\
&q(z)=-\frac{2 (z^2-1) \left(C-h^2+1\right)}{C(z^2+1)^2 \left[(h-1) z^2+h+1\right]}.
\end{split}
\label{eq:pqspec}
\end{equation}

If we make the following change of the independent variable
\[
z\to y=\frac{1}{z^2+1}+\frac{h-3}{6},
\]
then equation \eqref{eq:var} with coefficients \eqref{eq:pqspec}  transforms into
\begin{equation}
\frac{\mathrm{d}^2m}{\mathrm{d} y^2} +\frac{f'(y)}{2f(y)} \frac{\mathrm{d}m}{\mathrm{d} y}-\frac{\alpha y+\beta}{f(y)}m=0,
\label{eq:lamiak}
\end{equation}
with function $f(y)$
\begin{equation*}
 f(y)=4y^3-g_2y-g_3,\,\,\,   
 g_2=\frac{h^2}{3}+1,\,\,\, g_3=-\frac{1}{27} h \left(h^2-9\right), 
\end{equation*}
and parameters 
\begin{equation*}
\begin{split}
&\alpha= \frac{2\left(C-h^2+1\right)}{C}=2-\frac{2 A}{C},\\
&\beta=-\frac{h \left(C-h^2+1\right)}{3 C}=\frac{h (A-C)}{3 C}.
\end{split}
\end{equation*}
In equation \eqref{eq:lamiak} one can recognize the algebraic form of the Lam\'e
equation; see Appendix~\ref{sec:lamiaste}. The discriminant
$27g_3^2-g_2^3=-(h^2-1)^2=-A^2\neq0$. 

Usually, the Lam\'e equation is written
with the parameter $n$ instead of $\alpha$  defined as $\alpha=n(n+1)$ and for
our equation equals
\begin{equation}
n=\frac{1}{2} \left(-1\pm\sqrt{9-8\frac{A}{C}}\right).
\label{eq:lamiasten}
\end{equation}

Conditions which guarantee that the identity component of the differential
Galois group of the Lam\'e equation  is Abelian expressed in terms $n,\beta,g_2$
and $g_3$ are given  in  Lemma~\ref{lem:lame}.  

Let us check the values of $n$ corresponding to the Kovalevskaya and
Goryachev-Chaplygin cases. In the Kovalevskaya case $A=B=2C$,  and in this case
expression on $n$ \eqref{eq:lamiasten} in the corresponding Lam\'e equation is
\[
n=\frac{1}{2} \left(-1\pm\sqrt{7}\,\rmi\right).
\]
 In the Goryachev-Chaplygin case  $A=B=4C$ equation \eqref{eq:lamiasten} gives
\[
n=\frac{1}{2} \left(-1\pm\sqrt{23}\,\rmi\right).
\]
These values of $n$ do not agree with any form admissible in three cases when
the identity component of the differential Galois group has an Abelian identity
component mentioned in Lemma~\ref{lem:lame} in Appendix~\ref{sec:lamiaste}.
Moreover, as the covering $t\mapsto z$ given by~\eqref{eq:tran} does not change
the identity component of the differential Galois group of the normal
variational equations~\eqref{eq:rnvegneq0} is not Abelian and by
Theorem~\ref{thm:Morales} the Kovalevskaya and the Goryachev-Chaplygin cases of
our system are not integrable. This ends the proof. 
\end{proof}

\section{Integrability analysis. Case $g= 0$}
\label{sec:g0}

If $g=0$ and $\vL=\vzero$ then the system is integrable.  
In this section, we assume that $g=0$ and $\vL\neq\vzero $. 
Unlike for the case $g\neq0$, Lemma~\ref{lem:lowest} does not give any
obstruction for the integrability. This is why we cannot use the Ziglin
Theorem~\ref{thm:z1} and reduce our consideration to a symmetric body. 
In the following section, we prove the theorem that gives obstruction to the
integrability of the sliding top similar to those given by
Theorem~\ref{thm:z1} for the heavy top. 

\subsection{Non-symmetric case}
\label{sec:g0generic}
We will prove the following theorem.
\begin{theorem}
    \label{thm:g0gen}
If  $(A-B)(B- C)(C-A)\neq 0$, and $\vL\neq\vzero$, then the
system~\eqref{eq:poissonowska} with $g=0$ is not integrable  with complex meromorphic first integrals. 
\end{theorem}
\begin{proof}
In our proof, we will use Theorem~\ref{thm:Ziglin} mentioned in Section~\ref{sec:Ziglin}. 

At first, we set $\vL = \varepsilon \vS$, and consider $\varepsilon$ as
a small parameter. Then, the Hamiltonian of the system~\eqref{integralsM} can be
written as a perturbation of the integrable Euler top 
\begin{equation}
  H = H_0 
   + \varepsilon^2 H_1 + \cdots,
   \label{eq:HHDA}
\end{equation}
where 
\begin{equation}
\begin{split}
& H_0 = 
 \frac{1}{2} \left( \frac{M_1^2}{A} + \frac{M_2^2}{B}+\frac{M_3^2}{C} \right),\\ 
& H_1 = \frac{1}{2} m \left[\vGamma\cdot(\vOmega\times\vS)\right]^2, \quad
  \vM = \vI\vOmega.
  \end{split}
  \label{H0H1}
\end{equation}
Hence, we can apply the results of \cite{Ziglin:80::e}. To this
end we have to introduce the Androyer-Deprit canonical coordinates: angles
$(\al,\ag,\ah)$ and momenta $(\aL,\aG,\aH)$
 related to
$(\Gamma_1,\Gamma_2,\Gamma_3)$ and $(M_1,M_2,M_3)$ by formulae
\begin{widetext}
\begin{equation}
\begin{split} 
M_1=&\sqrt{\aG^2-\aL^2}\sin \al, \quad
M_2=\sqrt{\aG^2-\aL^2}\cos \al, \quad M_3=\aL,\\
\Gamma_1=&\left[\frac{\aH}{\aG}\sqrt{1-\frac{\aL^2}{\aG^2}}+\frac{\aL}{\aG}\sqrt{1-\frac{\aH^2}{\aG^2}}\cos \ag\right]\sin \al
+\sqrt{1-\frac{\aH^2}{\aG^2}}\sin \ag \cos \al, \\
\Gamma_2=&\left[\frac{\aH}{\aG}\sqrt{1-\frac{\aL^2}{\aG^2}}+\frac{\aL}{\aG}\sqrt{1-\frac{\aH^2}{\aG^2}} \cos \ag\right]\cos \al-\sqrt{1-\frac{\aH^2}{\aG^2}} \sin \ag \sin \al , \\
\Gamma_3=&\frac{\aL\aH}{\aG^2}-\sqrt{1-\frac{\aL^2}{\aG^2}}\sqrt{1-\frac{\aH^2}{\aG^2}} \cos \ag,
\end{split}
\label{eq:ADvariables}
\end{equation}
\end{widetext}
where $\aG>0$, $-\aG\leq\aL\leq\aG$ and $-\aG\leq\aH\leq\aG$, see
\cite{Andoyer:1923::,Deprit:67::,Kozlov:80::}. The inverse formulae are the
following
\begin{widetext}
\begin{equation*}
\begin{split}
&\aL=M_3,\quad \aG=\sqrt{\vM\cdot\vM},\quad \aH=\vM\cdot\vGamma,\quad
\al=\arctan\left(\frac{M_1}{M_2} \right),\quad \ag=\arctan\left(\frac{M_2\Gamma_1-M_1\Gamma_2}{\frac{M_3(\vM\cdot\vGamma)}{\sqrt{\vM\cdot\vM}}-\sqrt{\vM\cdot\vM}\Gamma_3}\right).
\end{split}
\end{equation*}
\end{widetext}
Let us notice that $\aH=F_2$ is constant for our system and the classical
heavy top problem. In these variables $H_0$ and $H_1$ take the forms
\begin{widetext}
\begin{equation*}
\begin{split}
&H_0=\frac{1}{2} \left((\aG^2-\aL^2)  \left(a \sin^2\al+b \cos^2\al\right)+c
   \aL^2\right),\\
&H_1=-\frac{m}{\aG^4} \Big[\aG \sqrt{\aG^2-\aH^2}\sin\ag  \Big[\sin\al \Big(c S_1\aL
-a S_3
   \sqrt{\aG^2-\aL^2} \sin\al\Big)-b S_3 \sqrt{\aG^2-\aL^2} \cos^2\al+c S_2 \aL \cos\al\Big]\\
   &+\sqrt{\aG^2-\aH^2}\cos\ag  \Big[\cos\al \Big(  (a-b)S_3\aL 
  \sqrt{\aG^2-\aL^2}
   \sin\al-S_1 \left(b \aG^2+(c-b) \aL^2\right)\Big)+S_2 \sin\al \left(a
   \aG^2+(c-a) \aL^2\right)\Big]\\
   &+\sqrt{\aG^2-\aL^2} \Big(\aH \cos\al 
   \Big[(a-b)S_3
    \sqrt{\aG^2-\aL^2} \sin\al+ (b-c)S_1 \aL \Big]+(c-a) S_2 \aL\aH    \sin\al\Big)\Big]^2,
\end{split}    
\end{equation*}
\end{widetext}
where $a=\tfrac{1}{A}$, $b=\tfrac{1}{B}$ and $c=\tfrac{1}{C}$. Later, we will
assume that $0<a<b<c$.  In these variables $H_0$ does not depend on $\aH$ and
$\ah$. Moreover, the expansion \eqref{eq:HHDA} has the form
\eqref{eq:HexpZiglin} with variables $q=\al$, $p=\aL$, $I=\aG$ and
$\varphi=\ag$.

The Fourier series of  $H_1$ with respect to variable $\ag$ contains only five terms  
\begin{equation*}
 H_1=h_0+h_1 +h_{-1} +h_2 +h_{-2}. 
\end{equation*}
In our further considerations we will need to know the explicit form   of two of them $h_1$ and $h_2$:
\begin{widetext}
\begin{equation*}
\begin{split}
&h_1= \frac{m \aH}{4 \aG^4} \sqrt{\aG^2-\aH^2} \sqrt{\aG^2-\aL^2} 
\Big[2 \aL \Big((a-c) S_2 \sin \al
+ (c-b) S_1 \cos \al\Big)
-(a-b)S_3 \sqrt{\aG^2-\aL^2} \sin (2 \al)\Big] 
\Big[(a-b)\\
&\cdot S_3 
   \sqrt{\aG^2-\aL^2} \Big(\aL \sin (2 \al)-\rmi \aG \cos (2 \al)\Big)  +\rmi (a+b)S_3 \aG 
   \sqrt{\aG^2-\aL^2}+2 \sin\al
   \Big(a S_2 (\aG^2-\aL^2)+c \aL (S_2 \aL-\rmi S_1 \aG )\Big)\\
 &  -2 \cos \al \Big(b S_1(\aG^2-\aL^2)+c  (S_1 \aL+\rmi S_2 \aG)\aL\Big)\Big]\rme^{\rmi \ag}, \\
&h_2=\frac{m (\aH^2-\aG^2) }{32 \aG^4}\Big[4 \cos (2 \al) \Big(a^2 (\aG^2-\aL^2) 
   \Big(S_2^2 (\aL^2-\aG^2)+S_3^2\aG^2\Big)+2 a c S_2\aL (\aL^2-\aG^2) (S_2\aL
-\rmi S_1 \aG)+b^2 (\aG^2\\ &-\aL^2) \left((S_1^2-S_3^2)\aG^2 -S_1^2\aL^2\right)+2 b c S_1\aL (\aG^2-\aL^2) (S_1\aL+\rmi S_2 \aG)
   +c^2 \aL^2 \left[(S_1^2-S_2^2) (\aG^2+\aL^2)
  +4 \rmi S_1 S_2\aG \aL\right]\Big)\\
  &+(\aG^2-\aL^2)
   \Big(S_3^2 \left((a-b)^2\aL^2- (3 a^2+2 a b+3
   b^2)\aG^2\right)+4 (-\rmi a S_2 (\aG+\aL)+b S_1 (\aG+\aL)-c (S_1-\rmi S_2)\aL) (\rmi S_2 (a (\aG\\ &-\aL)+c \aL)
  +b S_1
   (\aG-\aL)+c S_1\aL)\Big)-8 \rmi S_3 \sqrt{\aG^2-\aL^2} \Big(
   (a+b)\aG
   -(a-b) [\aG \cos (2 \al)+\rmi \aL \sin (2 \al)]\Big) \Big[\sin \al \Big(a
   S_2 (\aL^2\\ &-\aG^2)
   +c \aL (-S_2\aL+\rmi S_1\aG)\Big)+\cos \al \Big(b S_1 (\aG^2-\aL^2) +c \aL (S_1\aL
   +\rmi S_2\aG)\Big)\Big]
   -4 \sin (2 \al) \Big(2 \Big(a S_2
   (\aG^2-\aL^2)+c \aL (S_2\aL\\
   &-\rmi S_1\aG)\Big) \left(b S_1 (\aG^2-\aL^2)
   +c \aL (S_1\aL+\rmi S_2\aG)\right) +\rmi   (b^2-
   a^2) S_3^2 \aG \aL
   (\aG^2-\aL^2)\Big) -(a-b)^2S_3^2
   (\aG^4-\aL^4) \cos (4 \al)\\
 &  -2 \rmi S_3^2 (a-b)^2 \aG \aL(\aG^2-\aL^2) \sin (4 \al)\Big] \rme^{\rmi 2 \ag}.    
\end{split}   
\end{equation*}
\end{widetext}

For $\varepsilon=0$ our system reduces to an integrable Euler top with
Hamiltonian $H_0$. Thus, the formulae for solutions given in
\cite{Ziglin:80::e} can be applied directly. Equations \eqref{eq:nonperturb} for
the unperturbed Euler top  read
 \begin{equation}
 \begin{split}
\frac{\mathrm{d} \al}{\mathrm{d} t}&=
\frac{\partial H_0(\al,\aL,G_0 )}{\partial \aL}=\aL\left(c-(a\sin^2\al+b\cos^2\al)\right),\\
\frac{\mathrm{d} \aL}{\mathrm{d} t}&=
-\frac{\partial H_0(\al,\aL,G_0 )}{\partial \al}=(b-a)(G_0-\aL^2)\sin\al\cos\al,
\end{split}
\label{eq:nonperturbspec}
\end{equation}
where $\aG=G_0>0$. They  
 have two hyperbolic equilibria at $(\aL,\al)=(0,0)$ and $(\aL,\al)=(0,\pi)$,
 with
eigenvalues $\{-\lambda,\lambda\}$, where  $\lambda=G_0\sqrt{(b-a)(c-b)}>0$, see
 Fig.~\eqref{fig:topasym}.
These points are connected by two double-asymptotic solutions
\begin{equation}
\begin{split}
&\widehat{\aL}^{\pm}(t)= \pm 2G_0\sqrt{\frac{b-a}{c-a}}
\frac{e^{\lambda t}}{\left(1+e^{2\lambda t}\right)},\\
&\widehat{\al}^{\pm}(t)=\arctan\left( \pm\frac{1}{2}
\sqrt{\frac{c-a}{c-b}}\frac{\left(1-e^{2\lambda t}\right)}{e^{\lambda t}}\right).
\end{split}
\label{eq:solsEulertop}
\end{equation}
Function $\widehat{\ag}^{\pm}(t)$ for unperturbed system one can obtain
integrating the third equation in \eqref{eq:sysZigunpertur}
\begin{equation*}
\begin{split}
\widehat{\ag}^{\pm}(t)&=\int\frac{\partial H_0}{\partial \aG}(\widehat{\al}^{\pm}(t),\widehat{\aL}^{\pm}(t),G_0)\mathrm{d}t\\
&=G_0\int\left(a\sin^2\widehat{\al}^{\pm}(t)+b\cos^2\widehat{\al}^{\pm}(t)\right)\mathrm{d}t\\
&=G_0bt+\frac{\rmi}{2}\ln\left(\frac{e^{2\lambda t}-\gamma_1^2}{e^{2\lambda t}-\gamma_2^2} \right)+g_0,
\end{split}
\end{equation*}
where we introduced new parameters depending on $a$, $b$ and $c$
\[
\begin{split}
&\gamma_1=\sqrt{(b-a)(c-a)}+\rmi \sqrt{(c-b)(c-a)},\\
&\gamma_2=\gamma_1^{\ast}=\sqrt{(b-a)(c-a)}-\rmi \sqrt{(c-b)(c-a)},
\end{split}
\]
and symbol ${}^{\ast}$ denotes the complex conjugation, see \cite{Ziglin:80::e}. 
From \eqref{eq:solsEulertop} one deduce the following identities
\begin{equation}
\begin{split}
 &\sin\widehat{\al}^{\pm}(t) =2\sqrt{\frac{c-b}{c-a}}\frac{\rme^{\lambda t}}{\sqrt{(\rme^{2\lambda t}-\gamma_1^2)(e^{2\lambda t}-\gamma_2^2)}},\\
 & \cos\widehat{\al}^{\pm}(t) =\pm\frac{1-e^{2\lambda t}}{\sqrt{(e^{2\lambda t}-\gamma_1^2)(e^{2\lambda t}-\gamma_2^2)}},\\
  &\sqrt{1-\left(\frac{\widehat{\aL}^{\pm}(t)}{G_0}\right)^2}=\frac{\sqrt{(e^{2\lambda t}-\gamma_1^2)(e^{2\lambda t}-\gamma_2^2)}}{1+e^{2\lambda t}}.
\end{split}    
\end{equation}
Next, we evaluate $h_1$ and $h_2$ at these solutions, and the results are
denoted by $h_1^\pm(z)$ and $h_2^\pm(z)$ where $z=\rme^{\lambda t}$. The only
poles of these functions are at $z=\pm \rmi$, see formulae~\eqref{eq:h1z},
\eqref{eq:h1zm}, \eqref{eq:h2z} and \eqref{eq:h2zm}, respectively.  Let 
\begin{equation}
\scS^\pm_k =  \operatorname{res}\,(h_k^{\pm},-\rmi) + \operatorname{res}\,(h_k^{\pm},+\rmi), \qquad k=1,2, 
\end{equation}
be the sum of the residues of the function $h_k^{\pm}$.  These sums are of the form 
$\scS^\pm_k = \scA^\pm_k R^\pm_k$ where $\scA^\pm_k \neq 0$, see
formulae~\eqref{eq:h1z} and \eqref{eq:h2z}.
By Theorem~\ref{thm:Ziglin}, if the system is integrable then  $\scS^\pm_k=0$
for $k=1,2$, so $R^\pm_k=0$.
Thus,
\begin{equation}
\label{eq:resA0}
\begin{split}
&\Re  R_1^{+}+\Re  R_1^{-}=4\beta S_1[A_{11}S_2+A_{12}S_3]=0,\\
&\Im  R_1^{+}+\Im  R_1^{-}=4\alpha\beta S_1[A_{21}S_2+A_{22}S_3]=0,
\end{split}    
\end{equation}
where 
\begin{equation*}
\begin{split}
&A_{11}=\sqrt{\alpha ^2+\beta ^2} \left(\rme^{\frac{\pi  b}{\alpha  \beta }}-1\right) \left(\alpha ^2 \beta ^2-2 b^2+3 \alpha ^2 b\right),\\
&A_{12}=3 \alpha ^2 \beta  \left(\rme^{\frac{\pi  b}{\alpha  \beta }}+1\right) \left(-\alpha ^2+\beta ^2+b\right),\\
&A_{21}=-3 \beta  \sqrt{\alpha ^2+\beta ^2} \left(\alpha ^2-b\right) \left(\rme^{\frac{\pi  b}{\alpha  \beta }}-1\right),\\
&A_{22}=\left(\rme^{\frac{\pi  b}{\alpha  \beta }}+1\right) \left(\beta ^2 \left(3 b-4 \alpha ^2\right)+b \left(2 b-3 \alpha ^2\right)\right).
\end{split}    
\end{equation*}
In the above formulae $\alpha=\sqrt{b-a}$ and $\beta=\sqrt{c-b}$ and
$\alpha\beta\neq 0$. 

Let us assume that  $S_1\neq0$. Then 
  we can consider~\eqref{eq:resA0} as a system of linear
homogeneous equations for $S_2$ and $S_3$.  It has  a non-zero solution  if the
determinant
\begin{equation*}
\begin{split}
&\begin{vmatrix}  
A_{11}&A_{12}\\
A_{21}&A_{22}
\end{vmatrix}=\sqrt{\alpha ^2+\beta ^2} \left(1-\rme^{\frac{2 \pi  b}{\alpha  \beta }}\right)W_1,\\
&W_1=b^2 \left(2 b-3 \alpha ^2\right)^2+\beta^2 \left(9 \alpha ^6+6 b^3-10 \alpha ^2 b^2-3 \alpha ^4 b\right)\\
&+\alpha^2 \beta^4 \left(6 b-5 \alpha ^2\right),
\end{split}
\end{equation*}
vanishes. 

Similarly,  we get
\begin{equation*}
\begin{split}
&\Re  R_2^{+}-\Re  R_2^{-}=-8b\alpha\beta S_1[B_{11}S_2+B_{12}S_3],\\ 
&\Im  R_2^{+}-\Im  R_2^{-}=8\beta S_1 [B_{21}S_2+B_{22}S_3],
\end{split}    
\end{equation*}
where
\begin{equation*}
\begin{split}
&B_{11}=-3 \beta  \sqrt{\alpha ^2+\beta ^2} \left(\alpha ^2-b\right) \left(\rme^{\frac{2 \pi  b}{\alpha  \beta }}+1\right),\\
&B_{12}=\left(\rme^{\frac{2 \pi  b}{\alpha  \beta }}-1\right) \left(\beta ^2 \left(3 b-4 \alpha ^2\right)+b \left(2 b-3 \alpha ^2\right)\right),\\
&B_{21}=8 b \beta  S_1 \sqrt{\alpha ^2+\beta ^2} \left(\rme^{\frac{2 \pi  b}{\alpha  \beta }}+1\right) \left(\alpha ^2 \beta ^2-2 b^2+3 \alpha ^2 b\right),\\
&B_{22}=-12 \alpha ^2 \beta ^2 S_1 \left(\rme^{\frac{2 \pi  b}{\alpha  \beta }}-1\right) \left[\alpha ^2 \beta ^2-2 b^2+2 b (\alpha^2 -\beta^2 )\right].
\end{split}    
\end{equation*}
The determinant of the system for homogeneous system on $S_2$ and $S_3$ is
\begin{equation*}
\begin{split}
&\begin{vmatrix}  
B_{11}&B_{12}\\
B_{21}&B_{22}
\end{vmatrix}=4 \beta  S_1 \sqrt{\alpha ^2+\beta ^2} \left(\rme^{\frac{4 \pi  b}{\alpha  \beta }}-1\right)W_2,\\
&W_2=2 b^3 \left(2 b-3 \alpha ^2\right)^2+\alpha ^2 \beta ^4 \left(9 \alpha ^4+12 b^2-19 \alpha ^2 b\right)\\
&+2 b \beta ^2 \left(9 \alpha ^6+6 b^3-10 \alpha ^2 b^2-3 \alpha ^4 b\right).
   \end{split}
\end{equation*}
Hence, if the system is integrable, then  $W_1=W_2=0$. These equations have  the following solutions
\[
\begin{split}
&\{b= 0,\alpha = 0\},\quad\left\{b= -\frac{3 \beta ^2}{2},\alpha = 0\right\},\quad\{b= 0,\beta = 0\},\\
&\left\{b= \frac{3 \alpha ^2}{2},\beta =
   0\right\},\,\,\,
\left\{b= \alpha ^2,\beta = -\rmi \alpha \right\},\,\,\, \left\{b= \alpha ^2,\beta = \rmi \alpha \right\},
   \end{split}
\]
but all of them are excluded by assumptions. 
Summarizing,  if $S_1\neq 0$, then necessarily $S_2=S_3=0$. But if it is so,
then 
\begin{equation*}
  \begin{split}
    \Re R_1^{+}=&3 \alpha  \beta^3\left(b - 2 \alpha ^2\right) \left(e^{\frac{\pi  b}{\alpha  \beta }}+1\right) S_1^2 =0 ,\\
    \Re R_1^{+} = & \beta^2\left(b^2-2 \alpha ^2 \beta ^2-3 \alpha ^2 b\right) \left(e^{\frac{\pi  b}{\alpha  \beta }}+1\right) S_1^2=0.
  \end{split}
\end{equation*}
Hence the above equalities give
\begin{equation*}
  b = 2 \alpha ^2  \mtext{and} b^2-2 \alpha ^2 \beta ^2-3 \alpha ^2 b=0.
\end{equation*}
This implies that  $\alpha^2(\alpha^2+\beta^2)=0$, but it is impossible. 
We conclude $S_1\neq0$ is impossible. 

Thus, let $S_1=0$. We show, that in this case $S_2S_3\neq0$. In fact, 
if $S_1=S_3=0$ and $S_2\neq 0$,then 
\[
\begin{split}
&\Re R_1^{+}=-3 \alpha  b \beta  \left(\alpha ^2+\beta ^2\right) \left(e^{\frac{\pi  b}{\alpha  \beta }}+1\right)S_2^2\neq 0.
\end{split}
\]
Thus,  $S_3\neq 0$.  Similarly, if $S_1=S_2=0$ and $S_3\neq 0$,then 
\begin{equation*}
  \Re R_1^{+}=3 \alpha ^3 \beta \left(2 \beta ^2+b\right) \left(e^{\frac{\pi  b}{\alpha  \beta }}+1\right) S_3^2\neq 0.
\end{equation*}
Thus, we proved our claim.

We show that the last case $S_1 =0$ and $S_2S_3\neq0$ is impossible. 
 We consider combinations
\begin{equation}
\begin{split}
&2 \left(\alpha ^2 \beta ^2+b^2\right) \Re R_1^{+}  -3 \alpha  b \beta\Im  R_1^{+}\\
&=
2\alpha S_3[C_{11}S_2+C_{12}S_3]=0,\\
&3 \left(\alpha ^2 \beta ^2+2 b^2\right)\Re R_2^{+}+4 b \left(\alpha ^2 \beta ^2+b^2\right))
\Im R_2^{+}\\
&=4b\alpha S_3[C_{21}S_2+C_{22}S_3]=0,
\end{split}
\label{eq:linek1}
\end{equation}
where 
\[
\begin{split}
&C_{11}=  \sqrt{\alpha ^2+\beta ^2} \left(\rme^{\frac{\pi  b}{\alpha  \beta }}-1\right) \Big(-4 b^4-\beta ^2 b^2 \left(11 \alpha ^2+6 b\right)\\
&+\alpha ^2 \beta
   ^4 \left(2 \alpha ^2-15 b\right)\Big),\\
&C_{12}=  3 \alpha ^2 \beta ^3 \left(\rme^{\frac{\pi  b}{\alpha  \beta }}+1\right) \left(2 \alpha ^2 \beta ^2-b^2+3 \alpha ^2 b\right),\\
&C_{21}=-\sqrt{\alpha ^2+\beta ^2} \left(\rme^{\frac{2 \pi  b}{\alpha  \beta }}+1\right) \Big(9 \alpha ^4 \beta ^6+8 b^5\\
&+2 b^3 \beta ^2 \left(11 \alpha
   ^2+6 b\right)+5 \alpha ^2 b \beta ^4 \left(\alpha ^2+6 b\right)\Big),\\
 &C_{22}=3 \alpha ^2 \beta ^3 \left(\rme^{\frac{2 \pi  b}{\alpha  \beta }}-1\right) \left(\alpha ^2 \beta ^2 \left(\alpha ^2+b\right)-2 b^2 \left(b-2 \alpha
   ^2\right)\right).  
\end{split}
\]
Let us notice that the coefficients  of linear combinations on left-hand sides
of equation \eqref{eq:linek1} do not vanish. 
The determinant of matrix $[C_{ij}]$ is
\begin{equation*}
  \det[C_{ij}]=12 \alpha ^2 \beta ^3 \sqrt{\alpha ^2+\beta ^2}
     e^{\frac{3 \pi  b}{2 \alpha  \beta }} 
     \cosh \left(\frac{\pi  b}{2 \alpha  \beta }\right) 
     \left[ C_1 X +C_2  \right],
\end{equation*}
where $X=\cosh \left(\tfrac{\pi  b}{ \alpha  \beta }\right)$, and 
\begin{equation}
  \label{eq:C1C2}
  \begin{split}
    C_1=&
       \alpha ^2 \left(\alpha ^2 \beta ^2+b^2\right)^2 
       \Big(9 \beta ^4+4 b^2+\beta ^2 (\alpha ^2+12 b)\Big),\\
    C_2= &  -\left(2 b^2 \left(b-2 \alpha ^2\right)-\alpha ^2 \beta ^2 
       \left(\alpha ^2+b\right)\right) \Big(4 b^4\\
       &+\beta ^2 b^2
       \left(11 \alpha ^2+6 b\right)+\alpha ^2 \beta ^4 
       \left(15 b-2 \alpha ^2\right)\Big).
       \end{split}
\end{equation}

We consider also other combination
\begin{equation}
\begin{split}
& \hspace{-0.15cm}2 \left(b^2+\beta ^2 \left(3 b-2 \alpha ^2\right)\right)\Re R_1^{+} -3 \alpha  \beta  \left(2 \beta ^2+b\right) \Im  R_1^{+}\\
&=-2\alpha S_2[D_{11}S_2+D_{12}S_3]=0,\\
&3 \alpha  \beta  \left(2 b^2+\beta ^2 \left(4 b-\alpha ^2\right)\right)\Re R_2^{+}+4 b \Big(b^2+\beta ^2 (3 b\\
&-2 \alpha^2)\Big)\Im R_2^{+}=4\alpha bS_2[D_{21}S_2+D_{22}S_3]=0,
\end{split}
\label{eq:linek2}
\end{equation}
where
\[
\begin{split}
&D_{11}=3 \beta ^3 \left(\alpha ^2+\beta ^2\right) \left(\rme^{\frac{\pi  b}{\alpha  \beta }}+1\right) \left(-2 \alpha ^2 \beta ^2+b^2-3 \alpha ^2 b\right),\\
&D_{12}=\sqrt{\alpha ^2+\beta ^2} \left(\rme^{\frac{\pi  b}{\alpha  \beta }}-1\right) \Big[18 \alpha ^2 \beta ^6+4 b^4\\
&+b^2 \beta ^2 \left(18 b-\alpha
   ^2\right)+\beta ^4 \left(4 \alpha ^4+18 b^2+9 \alpha ^2 b\right)\Big],\\
&D_{21}=3 \beta ^3 \left(\alpha ^2+\beta ^2\right) \left(\rme^{\frac{2 \pi  b}{\alpha  \beta }}-1\right) \Big(\alpha ^2 \beta ^2 (\alpha ^2+b)\\
&-2
   b^2 (b-2 \alpha ^2)\Big),\\
&D_{22}=\sqrt{\alpha ^2+\beta ^2} \left(\rme^{\frac{2 \pi  b}{\alpha  \beta }}+1\right) \Big(-8 b^5+2 b^3 \beta ^2 (\alpha ^2-18 b)\\
&+b \beta ^4
   (\alpha ^4-36 b^2-18 \alpha ^2 b)+9 \alpha ^2
    \beta ^6 (\alpha ^2-4 b)\Big).
\end{split}
\]
As in the previous case, the linear combination coefficients on the right-hand
side of the equation \eqref{eq:linek2} do not vanish because $2 b^2+\beta ^2 \left(4
b-\alpha ^2\right)=a(c-b)+b(3c-b)>0$ and $b^2+\beta ^2 \left(3 b-2 \alpha
^2\right)=2a(c-b)+bc>0$.

The determinant of matrix $[D_{ij}]$ is
\begin{equation*}
\begin{split}
 \det[D_{ij}]=&-12 \beta ^3 \left(\alpha ^2+\beta ^2\right)^{3/2} 
\rme^{\frac{3 \pi  b}{2 \alpha  \beta }} \cosh \left(\frac{\pi  b}{2 \alpha  
\beta }\right)\\
&\cdot\left[D_1 X +D_2\right],\\
    D_1 = & 2 \alpha ^2 \left(\alpha ^2 \beta ^2+b^2\right) 
\left(2 \alpha ^2 \beta ^2-b^2-3 \beta ^2 b\right) \\
&\cdot\Big(9 \beta ^4
+4 b^2+\beta ^2 \left(\alpha
   ^2+12 b\right)\Big) , \\
D_2 = &
   \left(2 b^2 \left(b-2 \alpha ^2\right)-\alpha ^2 \beta ^2 \left(\alpha ^2+
   b\right)\right) \Big(18 \alpha ^2 \beta ^6+4 b^4\\
&+b^2 \beta ^2 (18 b-\alpha ^2)+ 
   \beta ^4 \left(4 \alpha ^4+18 b^2+9 \alpha ^2 b\right)\Big).
  \end{split}
\end{equation*}
Now, as $S_2S_3\neq$, equations~\eqref{eq:linek1} and~\eqref{eq:linek2} imply
that 
\begin{equation}
  \label{eq:CDX}
  C_1 X +C_2 = 0 \mtext{and} D_1 X +D_2=0
\end{equation}
and the necessary condition for this is  
\begin{multline*}
    C_1 D_2 - C_2D_1 =R (a b+2 a c-2 b c) \Big(a^2 b-a^2 c+a b^2\\+3 a b c-2 b^2 c\Big)=0, 
\end{multline*}
where $
   R =c (a-b)^2 (b-c)^2  [a b +c(b-a)-a c+b c]  \left[a b +c(9c -5b -a)\right]\neq 0
$, 
 because $0<a<b<c$. Hence,  we have only two possibilities 
\begin{equation}
  \label{eq:condory}
  \begin{split}
  &R_1=(a b+2 a c-2 b c) = 0, \mtext{or}\\
  &R_2= a^2 b-a^2 c+a b^2+3 a b c-2 b^2 c=0.
  \end{split}
\end{equation}
Solving the first condition for $c$ and substituting it into
equations~\eqref{eq:C1C2} we easily find that $b=2a$, then $C_1=C_2=0$. However,
it is impossible because if $b=2a$ then $R_1 = 2a(a-c)\neq0$. On the other hand,
if $b\neq2a$ then $ C_1\neq 0$ and $C_2\neq 0$, so equations~\eqref{eq:CDX}
reduces to  $X=1$, but it is impossible because by assumptions $X>1$.

Condition $R_2 = 0$ can be solved for $c$
\begin{equation}
  c=\frac{a b (a+b)}{a^2-3 a b+2 b^2}.
\end{equation}
One can check that with this $c$ we have $C_2=D_2=0$. Thus, as $X\neq 0$, we
must have $C_1=0$ and $D_1$. It is easy to check that this is possible only when
$7a^2 - b^2=0$. However, if $a=b/\sqrt{7}$ then 
\begin{equation*}
 R_2= \frac{1}{7} b^2 \left[\left(1+\sqrt{7}\right) b+3 \left(\sqrt{7}-5\right) c\right]=0,
\end{equation*}
but this is impossible because $c>b$. 

Summarizing assumptions that the system is integrable, and the body is not
symmetric leads to a contradiction. This finishes the proof. 
\end{proof}

Following \cite{Ziglin:80::e} one can also check the conditions of vanishing separately residues of the function $h_1(z)$ at $\pm\rmi$. Calculations yield that that conditions $\operatorname{res}\,(h_1^{+},\pm\rmi)=0$ give  $S_2=0$ and $\beta  S_1+\alpha  S_3=0$, and   $\operatorname{res}\,(h_1^{-},\pm\rmi)=0$ lead to $S_2=0$ and $ \alpha S_3 - \beta S_1=0$.
In parameters $a,b$ and $c$ these conditions take the forms
\begin{equation}
S_2=0\quad\text{and}\quad \sqrt{b-a} S_3\pm \sqrt{c-b}S_1=0,
 \label{eq:hess1}
\end{equation}
 compare with \eqref{eq:HAconditions}, which are the same as for the Hess-Appelrot case in the classical heavy top problem. In this case, like in the case of the heavy top,  our system is not integrable but possesses a Darboux polynomial \eqref{eq:HAdarbus} and its zero level is invariant with respect to the system.

\subsection{Generic symmetric case}
\label{ssec:g0non}
By Theorem~\ref{thm:g0gen}, if the system is integrable, then the body is
symmetric. Thus,  without loss of the generality we assume that $A=B$ and making
a rotation around the symmetry axis, we can achieve that the second component of
the vector $\vL$ in the principal axes frame vanishes. This guarantees that we
have a family of particular solutions needed for our non-integrability proof. 

 It is convenient to choose the body frame so that its first axis coincides with $\vL$ and the third axis of the body frame is perpendicular to the second principal axis. In this special coordinate frame  we have
\begin{equation}
\vL=[1,0,0]  ^T,\quad   \vI=\begin{bmatrix}
    a&0&2d\\
    0&b&0\\
    2d&0&c
\end{bmatrix},
\end{equation}
In the principal axes $\widetilde{\vL}=\vR[1,0,0]^T=[\widetilde{L}_1,0, \widetilde{L}_3]^T$,
and $\widetilde{\vI} = \vR\vI\vR^T = \diag(A,A,C)$, where 
\begin{equation}
  \vR = \begin{bmatrix}
    \widetilde{L}_1 & 0 & -\widetilde{L}_3\\
    0 & 1 & 0\\
    \widetilde{L}_3 & 0 & \widetilde{L}_1
  \end{bmatrix}.
\end{equation}
Hence, we have the following relations
\begin{equation*}
\begin{split}
&a=A\widetilde{L}_1^2 +C\widetilde{L}_3^2,\quad b=B=A,\quad 
c= C\widetilde{L}_1^2 +A\widetilde{L}_3^2,\\
&2d=(C-A)\widetilde{L}_1\widetilde{L}_3.
\end{split}
\end{equation*}
See \cite{mp:05::a} where a similar special body frame was used. 
The fact that $\vI$ is positively defined implies that 
\begin{equation}
  \label{eq:Ipos}
  a>0, \qquad b>0, \qquad (ac-4d^2)>0.
\end{equation}

Considering equations of motion in the introduced spacial body frame  once more we use the invariant manifold 
\begin{equation*}
 \cN:= \{ ( \vM,\vGamma)\in \C^6\ | \ M_1=M_3=\Gamma_2 = 0,
 \ \Gamma_1^2+\Gamma_3^2=1\}.  
 \label{eq:inva1geq0}
\end{equation*}
 The system restricted to $ \cN$ takes the form
\begin{equation}
  \label{eq:repgeq0}
  \frac{\mathrm{d} M_2}{\mathrm{d} t} = \frac{M_2^2\Gamma_1\Gamma_3}{(b+\Gamma_3^2)^2}, \,\,\,
  \frac{\mathrm{d} \Gamma_1}{\mathrm{d} t}=-\frac{M_2\Gamma_3}{b+\Gamma_3^2}, \,\,\, 
 \frac{\mathrm{d} \Gamma_3}{\mathrm{d} t} =\frac{M_2\Gamma_1}{b+\Gamma_3^2},
\end{equation}
and it gives a family of particular solutions obtained as the intersection of
two algebraic curves
\begin{equation}
h=\frac{M_2^2}{2 (b+\Gamma_3^2)},\quad \Gamma_1^2+\Gamma_3^2=1.  
\label{eq:partgeq0}
\end{equation}

For this particular solutions  the variational equations take the form
\begin{equation}
  \label{eq:varallvd0}
  \Dt \begin{bmatrix}
    m_1\\
    m_2\\
    m_3 \\
    \gamma_1 \\
    \gamma_2\\
    \gamma_3
  \end{bmatrix}
  =
\begin{bmatrix}
  a_{11} & 0 & a_{13} & 0 & a_{15} & 0 \\
  0 & a_{22} & 0 & a_{24} & 0 & a_{26}\\
  a_{31} & 0 & a_{33} & 0 & a_{35} & 0 \\
  0 & a_{42} & 0 & 0 & 0 & a_{46}\\
  a_{51} & 0 & a_{53} & 0 & a_{55} & 0 \\  
  0 & a_{62} & 0 & a_{64} & 0 & a_{66}\\
\end{bmatrix}
\begin{bmatrix}
  m_1\\
  m_2\\
  m_3 \\
  \gamma_1 \\
  \gamma_2\\
  \gamma_3
\end{bmatrix}.
\end{equation}

The normal variational equations are those corresponding to changes of variables $M_1,M_3$ and $\Gamma_2$
\begin{widetext}
\begin{equation}
\begin{split}
 \dot m_1 =&  -\frac{2 b d M_2}{(a c - 4 d^2) (b + \Gamma_3^2}m_1 +\frac{(a (b - c) + 4 d^2) M_2}{(a c - 4 d^2) (b + \Gamma_3^2)}m_3
 + \frac{(a (b - c) + 4 d^2) M_2^2 \Gamma_3}{(a c - 4 d^2) (b + \Gamma_3^2)^2}\gamma_2,\\
 \dot m_3=&  \frac{M_2 (a c - 2 d (2 d + \Gamma_1 \Gamma_3) - 
   c (b + \Gamma_3^2))}{(a c - 4 d^2) (b + \Gamma_3^2)}m_1
   +\frac{M_2 (a \Gamma_1 \Gamma_3 + 2 d (b + \Gamma_3^2)}{(a c - 4 d^2) (b + \Gamma_3^2)}m_3
   +\frac{M_2^2 \Gamma_3 (a \Gamma_1 \Gamma_3 + 
   2 d (b + \Gamma_3^2))}{(a c - 4 d^2) (b + \Gamma_3^2)^2}\gamma_2, \\
 \dot\gamma_2=& \frac{2 d \Gamma_1 + c \Gamma_3}{a c - 4 d^2}m_1-\frac{a \Gamma_1 + 2 d \Gamma_3}{a c - 4 d^2}m_3
 -\frac{M_2 \Gamma_3 (a \Gamma_1 + 2 d \Gamma_3)}{(a c - 4 d^2) (b + \Gamma_3^2)}\gamma_2.
\end{split}
\label{eq:norvariat}
\end{equation}
\end{widetext}
We use the zero level of the first integral of the normal variational equations
\[
 \Gamma_1 m_1   +  \Gamma_3m_3 + 
 M_2 \gamma_2=0.
\]
Using this constraint we eliminate   $\gamma_2$ from the normal variational
equations \eqref{eq:norvariat} and write the result as a second-order equation 
\begin{equation}
  \label{eq:rnvegeq0}
   \Dtt m + a_1(t) \Dt m + a_0(t) m = 0, \qquad m \equiv m_1,
\end{equation}
with coefficients
\[
\begin{split}
a_1(t)&=\frac{2 M_2 \Gamma_1\Gamma_3}{\left(b+\Gamma_3^2\right)^2},\\
a_0(t)&=\frac{M_2^2}{\left(a c-4 d^2\right)\left(b+\Gamma_3^2\right)^3}\Big[\Gamma_1^2 \left(a
   (b-c)+4 d^2\right)\\
   &-\left(b+\Gamma_3^2\right)
   \left((a-b) (b-c)+4 d^2\right)+2 b \Gamma_1
   \Gamma_3 d\Big].
\end{split}
\]
Now, we make the following transformation of the independent variable
\begin{equation}
\label{eq:trangeq0}
t \longrightarrow z := \frac{\Gamma_3(t)}{1+\Gamma_1(t)}.
\end{equation}
Then we obtain
\begin{gather*}
\Gamma_1 = \frac{1-z^2}{1+z^2}, \,\,\, \Gamma_3 = \frac{2 z}{1+z^2},   \,\,\, 
 M_2 = \frac{2\dot z(b + 4 z^2 + 2 b z^2 + b z^4)}{(1+z^2)^3}.
\end{gather*}
Using the first equation in \eqref{eq:partgeq0} and \eqref{eq:repgeq0} one can calculate derivatives
\[
\begin{split}
&\dot z^2=\frac{h (z^2+1)^4}{2 \left(b (z^2+1)^2+4
   z^2\right)},\\
&\ddot z=\frac{h z (z^2+1)^3 \left(b (z^2+1)^2+6
   z^2-2\right)}{\left(b (z^2+1)^2+4 z^2\right)^2}.
   \end{split}
\]
Now one can transform equation \eqref{eq:rnvegeq0} into
\begin{equation}
\label{eq:vargeq0}
m''+p(z)m'+q(z)m=0,\qquad '=\dfrac{\mathrm{d}}{\mathrm{d}z},
\end{equation}
with coefficients  
\begin{equation*}
\begin{split}
&p(z)=\frac{2 z \left(b z^2+b+2\right)}{b (z^2+1)^2+4 z^2},\\
&q(z)=-\frac{4}{ac-4d^2}\Big[\frac{(2 a-b) (b-c)+2 b d z+8 d^2}{(z^2+1)^2}+\frac{b d z}{z^2+1}\\
&-\frac{a (b+1)
   (b-c)+d \left(4 (b+1) d+b z \left(b
   \left(z^2+3\right)+4\right)\right)}{b (z^2+1)^2+4
   z^2}
\Big].
\end{split}
\label{eq:ppqqgeq0}
\end{equation*}
Making the following change of the dependent  variable
\begin{equation*}
m=w\exp\left[-\dfrac{1}{2}\int_{z_0}^z p(\zeta)\mathrm{d}\zeta\right],
\end{equation*}
one can simplify \eqref{eq:vargeq0} to the standard reduced form
\begin{equation}
\label{eq:zredgeq0}
w''=r(z)w,\qquad r(z)=\dfrac{1}{2}p'(z)+\dfrac{1}{4}p(z)^2-q(z). 
\end{equation}
The explicit form of $r(z)$  is  
\begin{equation}
\begin{split}
&\hspace{-0.4cm} r(z)=\frac{4 \left((2 a-b) (b-c)+2 b d
   z+8 d^2\right)}{\left(a c-4 d^2\right)(z^2+1)^2}\\
   &\hspace{-0.4cm}+\frac{4 b d z}{\left(a c-4 d^2\right)(z^2+1)}
   -\frac{12
   (b+1) z^2}{\left(b (z^2+1)^2+4 z^2\right)^2}\\
& \hspace{-0.4cm}+\frac{1}{\left(b (z^2+1)^2+4
   z^2\right) \left(a c-4 d^2\right)}\Big[a (5b+6) c\\
   &\hspace{-0.4cm}-4 a b (b+1)-4 d \left((5 b+6) d+b z \left(b
   \left(z^2+3\right)+4\right)\right)\Big].  
\end{split}
\label{eq:rrgeq0dneq0}
\end{equation}
Equation \eqref{eq:zredgeq0} has now generically seven singularities: six singularities 
$z_{1,2}=\pm\rmi$, $z_{3,4}=\pm \sqrt{-\tfrac{b+2 \sqrt{b+1}+2}{b}}$,
$z_{5,6}=\pm\sqrt{-\tfrac{b-2 \sqrt{b+1}+2}{b}}$ which are poles of the second order provided
\begin{equation*}
\begin{split}
&b(1+b)(a c - 4 d^2)U\neq0,\\
&U= (b-2 a)^2 (b - c)^2 + 4 (8 a (b - c) + b (-3 b + 4 c)) d^2 \\
&+ 64 d^4,
\label{eq:uuu}
\end{split}
\end{equation*}
and $z_{7}=\infty$.
We show that the above inequality always holds. In  fact, by assumptions
$b(1+b)(a c - 4 d^2)\neq0$. To show that $U\neq 0$  we introduce the
ratios of the parameters $x=\tfrac{b}{a}$, $y=\tfrac{c}{a}$ and
$z=\tfrac{d}{a}$. Then 
\[
\begin{split}
u&=\frac{U}{a^4}=(x-2)^2y^2-2 (x-2) \left((x-2) x-8 z^2\right)y\\
&+(x-2)^2 x^2+4 (8-3 x) x z^2+64 z^4.
\end{split}
\]
This polynomial treated as a quadratic polynomial of the variable $y$  has the discriminant
$\Delta=-16 (x-2)^2 x^2 z^2$.  Therefore, if $x\neq 2$  its roots are not real
but $y\in\R$ and this case is impossible. If $x=2$, then the polynomial $u$
simplifies to $u=16 z^2 (1 + 4 z^2)>0$ and does not vanish for real $z$.  Thus,
our claim is proved. 

In the generic case coefficient $r(z)$ in \eqref{eq:rrgeq0dneq0} has the following expansion
\begin{equation}
\label{eq:rozgeq0}
r(z)=\sum_{i=1}^6\left[\dfrac{\alpha_i}{(z-z_i)^2}+\dfrac{\beta_i}{z-z_i}\right],
\end{equation}
compare with formula \eqref{eq:r}, with coefficients $\alpha_i$
\[
\begin{split}
&\alpha_2 = \alpha_1^* = 2+\frac{b (-2 a+b-c)}{a c-4 d^2}+\rmi\frac{2 b d}{a c-4 d^2}, \\ 
&\alpha_3=\alpha_4=\alpha_5=\alpha_6=-\frac{3}{16} 
\end{split}
\]
and we do not write explicit formulas for $\beta_i$ because their forms do not play a role in our considerations. The differences of exponents $\Delta_i=\sqrt{1+4\alpha_i}$ at singular points $z_i$  equal
\begin{equation}
\label{eq:Deltaigeq0}
\begin{split}
&\Delta_1 = \Delta_2^* = \sqrt{9+\frac{4 b (-2 a+b-c)}{a c-4 d^2}-\rmi\frac{8 b d}{a c-4 d^2}}, \\
&\Delta_3=\Delta_4=\Delta_5=\Delta_6=\frac{1}{2}.
\end{split}
\end{equation}
The order of infinity is generically 4 provided $b V\neq0$
\[
V=4 a b(b-1) - 3 a (b-2) c + 4 b^2 (c-b) + 12 (b-2) d^2.
\]
If parameters are such that $V=0$, the infinity remains regular, so this does
not influence further analysis with the help of the Kovacic algorithm.

 The main result of this section can be formulated in the following way.
\begin{lemma}
\label{lem:2}
Let us assume that  $d\neq 0$. Then the identity component of the differential
Galois group of equation~\eqref{eq:zredgeq0} with coefficient $r(z)$ given in
\eqref{eq:rrgeq0dneq0} is not Abelian.
\end{lemma}
\begin{proof}
 Let $\scG$ denote the differential Galois group of equation~\eqref{eq:zredgeq0}.
  If its identity component $\scG^0$ is commutative, then either 
  \begin{enumerate}
    \item $\scG$ is a proper subgroup of triangular group $\scT$, see Lemma~\ref{lem:algc1} or
    \item $\scG$ is a subgroup of  $\scD^\dag$ defined in~\eqref{eq:Ddag}, or 
    \item $\scG$  is a finite group,
  \end{enumerate}
  see Appendix~\ref{sec:2ndorderlde}. 
  
In our proof we use the lemma~\ref{lem:alg} and~\ref{lem:algc1} from
Appendix~\ref{sec:2ndorderlde}. If equation~\eqref{eq:zredgeq0} is reducible and the
identity component $\scG^0$ of its differential Galois group is Abelian, then
$\scG$ is a subgroup of the diagonal group $\scD$, or $\scG$ is a proper
subgroup of triangular group $\scT$.  
  
  First, we show that $\scG\not\subset\scD$. Let us assume the opposite. Then there exist two exponential solutions of~\eqref{eq:zredgeq0} that have the following form.
\begin{equation}
\label{eq:expsol}
 w_l = P_l \prod_{i=1}^6(z-z_i)^{e_{i,l}}, \qquad P_l\in\C[z], \quad l=1,2, 
\end{equation}
where $e_{i,l}$ for $l=1,2$ are exponents at singular points $z_i$
\[
e_{i,l}\in\left\{ \frac{1}{2}(1 + \Delta_i), \frac{1}{2}(1 -\Delta_i)\right\}.
\]
Here $\Delta_i$ for $i=1,\ldots, 6$ are given by \eqref{eq:Deltaigeq0}. The product
of these solutions $v=w_1w_2$ is a rational function and is a solution of
the second symmetric power of equation~\eqref{eq:zredgeq0}, that is,
equation~\eqref{eq:ssp} with $r$ given by~\eqref{eq:roz}. This equation has the
same singular points as equation~\eqref{eq:zredgeq0}. Exponents $\rho_{i,l}$ at
singular points $z_i$, and at infinity $\rho_{\infty,l}$ for the second
symmetric power of~\eqref{eq:zredgeq0} are given by
\[
\rho_{i,l}\in  \{1, 1\pm\Delta_i\}, \quad 
\rho_{\infty,l}\in\{-2,-1,0\}, \quad l=1,2,3. 
\]
If we write $v = P/Q$ with $P,Q \in \C[z]$ then
\[
Q = \prod_{i=1}^6(z-r_i)^{n_i}, \qquad n_i\in\N, \qquad
r_i\in\{z_1,\ldots ,z_6\},
\]
and $ n_i= -\rho_{i,l}\in\N$ for certain $l$. However, if $d\neq0$,
then $\rho_{i,l}$ is not a negative integer for $i=1,\ldots, 6$ and
$l=1,2,3$. This implies that $Q=1$. Hence, equation~\eqref{eq:ssp} has
a polynomial solution $v=P$, and $\deg P = -\rho_{\infty,l}\leq 2$.
But $v$ is a product of two exponential solutions of the form
~\eqref{eq:expsol}, so we also have
\begin{equation}
\label{eq:exprod}
v = P_1P_2 \prod_{i=1}^6(z-z_i)^{e_{i,m}+e_{i,l}}\in\C[z], 
\quad m,l\in\{1,2\}.
\end{equation}
Consequently, ${e_{i,m}+e_{i,l}}$ is a non-negative integer for
$i=1,\ldots, 6$.  As for $d\neq0$, we have $2e_{i,l}\not\in\Z$ for
$i=1,\ldots, 6$ and $l=1,2$, we deduce that in~\eqref{eq:exprod} $m\neq
l$.  But $e_{i,1}+e_{i,2}=1$, for $i=1,\ldots, 6$.
Thus, we have 
\[ 
\deg v = \deg P =
\deg(P_1P_2)+6 > 2.
\]
We have a contradiction because we already showed that $\deg P \leq 2$.

It is also impossible that $\scG$ conjugates to $\scT_m$ for a certain
$m\in\N$ because when $d\neq 0$ exponents for $z_1$ and $z_2$ are not
rational. This implies also that $\scG$ is not finite.

The last possibility that $\scG^0$ is Abelian occurs when $\scG$ is
conjugated with a subgroup of $D^\dag$. We show that it is impossible.
To this end, we apply the second case of the Kovacic algorithm~\cite{Kovacic:86::}, see
Appendix~\ref{sec:2ndorderlde}. The auxiliary sets for singular points are the following
\[
\begin{split}
&E_1=E_2=\{2\}, \quad E_3=E_4 = E_5=E_6=\{1,2,3\}, \\
&E_\infty=\{0,2,4\}.
\end{split}
\]
In the Cartesian product $E=E_\infty\times E_1 \times \cdots \times E_6$ 
we look for such  elements $e$ for which 
\[
d(e):=\frac{1}{2}\left( e_\infty - \sum_{i=1}^6 e_i\right),
\]
is a non-negative integer, but there is no such element. Thus, the differential
Galois group of equation~\eqref{eq:zredgeq0} with coefficient $r(z)$ given in
\eqref{eq:rrgeq0dneq0} cannot be a subgroup of $D^\dag$ and its identity component is not Abelian.
\end{proof}
As the transformation  $t\mapsto z$ given by~\eqref{eq:trangeq0} does not change
the identity component of the differential Galois group of the normal
variational equations~\eqref{eq:rnve}, from the above lemma and
Theorem~~\ref{thm:Morales} it follows that if the Euler-Poisson equations are
integrable, then
\[ 
2 d = ( C-A){\widetilde L}_1 {\widetilde L}_3 =0.
\]

\subsection{Special symmetric case with $d=0$}
In this section we continue our integrability analysis assuming that
the necessary integrability condition $d=0$ formulated in Lemma~\ref{lem:3} is
satisfied and we use the condition that $B=A$. Vanishing of $d$ implies three
possibilities:
\begin{itemize}
\item $C=A$ that together with $B=A$ leads to fully symmetric integrable case $A=B=C$;
\item ${\widetilde L}_1=0$ that gives conditions ${\widetilde L}_1={\widetilde L}_2=0$ and $B=A$, thus integrable Lagrange case;
\item ${\widetilde L}_3=0$ and this case requires further analysis.
\end{itemize}

If ${\widetilde L}_3=0$, then $a=A$, $c=C$ and we have $\vI=\diag(A,A,C)$. 
In this case is useful to consider other  invariant manifold
\begin{equation*}
 \cM:= \{ ( \vM,\vGamma)\in \C^6\ | \ M_2=M_3=\Gamma_1 = 0,
 \ \Gamma_1^2+\Gamma_3^2=1\},   
\end{equation*}
different than in the case $d\neq0$.
The system restricted to $ \cM$ takes the form
\begin{equation}
  \label{eq:partsolgeqdeq0}
  \Dt M_1 = 0, \quad 
  \Dt \Gamma_2=\frac{M_1\Gamma_3}{A}, \quad 
  \Dt \Gamma_3=-\frac{M_1\Gamma_2}{A},
\end{equation}
and it gives a  particular solution.
The variational equations for this particular solution are following 
\begin{equation}
  \label{eq:varallvd00}
  \Dt \begin{bmatrix}
    m_1\\
    m_2\\
    m_3 \\
    \gamma_1 \\
    \gamma_2\\
    \gamma_3
  \end{bmatrix}
  =
\begin{bmatrix}
  0 & 0 & 0 & 0 & 0 & 0 \\
  0 & a_{22} & a_{23} & 0 & 0 & 0\\
  0 & a_{32} & a_{33} & 0 & 0 & 0 \\
  0 & a_{42} & a_{43} & 0 & 0 & 0\\
  a_{51} & 0 & 0 & 0 & 0 & a_{56}\\  
  a_{61} & 0 & 0 & 0 & a_{65} & 0
\end{bmatrix}
\begin{bmatrix}
  m_1\\
  m_2\\
  m_3 \\
  \gamma_1 \\
  \gamma_2\\
  \gamma_3
\end{bmatrix}.
\end{equation}

Now, the normal variational equations contain a subsystem for variables
$(m_2,m_3)$. 

We can rewrite it as one second-order  equation 
\begin{equation}
  \label{eq:rnve}
   \Dtt m + a_1(t) \Dt m + a_0(t) m = 0, \qquad m \equiv m_2,
\end{equation}
with coefficients
\[
\begin{split}
&a_1=\frac{2 A \Gamma_2\Gamma_3
   M_1 \left(A+\Gamma_2^2+\Gamma_3^2\right)}{\left[A
   \left(C+\Gamma_2^2\right)+C \Gamma_3^2\right] \left[(A-C) \left(A+\Gamma_3^2\right)-A \Gamma_2^2\right]},\\
&a_0= \frac{M_1^2}{A^2
   \left[A \left(C+\Gamma_2^2\right)+C\Gamma_3^2\right] \left[(A-C)
   \left(A+\Gamma_3^2\right)-A \Gamma_2^2\right]} \\
   &\cdot\Big[A^2 \Gamma_2^4+A
  \Gamma_2^2 \left(\Gamma_3^2 (3
   A-2 C)+A (C-A)\right)\\
   &+C \Gamma_3^2 (A-C)
   \left(A+\Gamma_3^2\right)\Big].
\end{split}
\]
Let us notice that system \eqref{eq:partsolgeqdeq0} has the solution
\begin{equation}
 M_1=const,\,\,\, \Gamma_2=\sin\left(\frac{M_1}{A} t\right),\,\,\,   
 \Gamma_3=\cos\left(\frac{M_1}{A} t\right).
 \label{eqparta}
\end{equation}
To rationalize normal variational equation \eqref{eq:rnve} we use the change of independent variable
\begin{equation}
\label{eq:trantryg}
t \longrightarrow z :=  \Gamma_3^2=\cos^2\left(\frac{M_1}{A} t\right)
\end{equation}
and  derivatives according to formulas
\[
\begin{split}
&\dot z^2= -\frac{4 M_1^2 (z-1) z}{A^2},\quad 
\ddot z=\frac{2 M_1^2 (1-2 z)}{A^2}.
\end{split}
\] 
Thus, we can transform equation \eqref{eq:rnve} into the following equation 
\begin{equation}
\label{eq:vargeq00}
m''+p(z)m'+q(z)m=0,\qquad '=\dfrac{\mathrm{d}}{\mathrm{d}z},
\end{equation}
with rational coefficients
\begin{equation*}
\begin{split}
 &p= \frac{\ddot z+a_1\dot z}{\dot z^2} =\frac{1}{2 z}+\frac{1}{2
   (z-1)}+\frac{A-C}{z (A-C)-A (C+1)}\\
   &+\frac{C-2 A}{z (2 A-C)+A (A-C-1)}
 ,\\
 &q=-\frac{1}{4 (C+1) z}-\frac{1}{4 (A+1)
   (z-1)}\\
   &+\frac{(A-2 C-2) (A-C)^2}{4 A (A+1) (C+1) [z (A-c)-A (C+1)]}\\
   &+\frac{(2
   A-C)^2}{2 A (A+1) [z (2 A-C)+A (A-C-1)]}.   
\end{split}
\end{equation*}

The coefficient $r(z)$ of the standard reduced form of normal variational equation \eqref{eq:zredgeq0} has the form
\begin{widetext}
\begin{equation}
\begin{split}
&r(z)=-\frac{3}{16 z^2}-\frac{3}{16 (z-1)^2}-\frac{1}{4
   \left(z-\frac{A (C+1)}{A-C}\right)^2}
   +\frac{3}{4 \left(z-\frac{A
   (C-A+1)}{2 A-C}\right)^2}
   -\frac{2 A^2+(A+3) C^2-(A+3) A C}{8
   (A+1) C (A-C)(z-1)}\\
 &  +\frac{-A (C+1)+C^2-3}{8 (C+1) (A-C-1)z} -\frac{(A-C)^2\left[A^2 (C+1)-A C (2 C+1)+2 c^2 (C+1)\right]}{4 A^2
   (A+1) C (C+1) [A (1 + C) + (C-A) z]}\\
   &+\frac{(C-2 A)^2 \left[2 A^3-4 A^2 C+A C (4
   C+1)-2 C^2 (C+1)\right]}{4 A^2 (A+1) (A-C-1) (A-C) \left[A (A - C-1) + (2 A - C) z\right]}.
\end{split}    
\label{eq:rrdgeq0}
\end{equation}
\end{widetext}
With the above coefficient $r(z)$,  equation \eqref{eq:zredgeq0}  has generically 
five singularities: four poles of the second order $z_1=0$, $z_2=1$, $z_3=\tfrac{A (C+1)}{A-C}$, $z_4=\tfrac{A (C-A+1)}{2A-C}$  provided
\begin{equation*}
AC(A+1) (C+1)(C-A)(1 - A + C)(2A - C)\neq0
\end{equation*}
and  $z_5=\infty$. The degree of infinity is 3 provided $A (A - C) (2A - C)\neq0$.

Now we will analyze the reduced form of normal variational equations
\eqref{eq:zredgeq0} with the above coefficient $r(z)$ in the generic case as
well as in nongeneric cases corresponding to the coalescence of singularities
or the change of the order of $r(z)$ at infinity. The results of this analysis
are formulated in the following lemma.  
\begin{lemma}
\label{lem:3}
Let us assume that $\vI=\diag(A,A,C)$ and ${\widetilde L}_2={\widetilde L}_3=0$. Then the identity component of the
differential Galois group of equation~\eqref{eq:zredgeq0} with coefficient $r(z)$ given in \eqref{eq:rrdgeq0} is not Abelian except the completely symmetric case $C=A$ which is a special case of the Lagrange case.
\end{lemma}

\begin{proof}
In the generic case coefficient $r(z)$ has expansion
\begin{equation}
\label{eq:roz}
r(z)=\sum_{i=1}^4\left[\dfrac{\alpha_i}{(z-z_i)^2}+\dfrac{\beta_i}{z-z_i}\right],
\end{equation}
with coefficients
\begin{gather*}
\alpha_1 = \alpha_2=-\frac{3}{16} \quad \alpha_3=-\frac{1}{4},\quad \alpha_4=\frac{3}{4}, 
\end{gather*}
and we do not write explicit formulas for $\beta_i$ because their form is
irrelevant at this point. 

The differences of exponents $\Delta_i=\sqrt{1+4a_i}$ at 
singular points $z_i$ are the following
\begin{equation*}
\Delta_1 = \Delta_2=\frac{1}{2},\quad \Delta_3=0,\quad \Delta_4=2.
\end{equation*}
Since $\Delta_3=0$, thus in local solutions around singularity $z_3$ logarithmic terms always appear and the differential Galois groups can be either
reducible or full $\mathrm{SL}(2,\C)$, see for details~\cite{mp:02::d}. To check
whether the group is reducible, we use the first case of the Kovacic algorithm,
see the appendix~\ref{sec:2ndorderlde}. The auxiliary sets of exponents at
singularities
\begin{equation*}
\begin{split}
 &E_1=E_2=\left\{\frac{1}{4},\frac{3}{4}\right\},\,\,\, E_3=\left\{\frac{1}{2},\frac{1}{2}\right\},\,\,\, E_4=\left\{-\frac{1}{2},\frac{3}{2}\right\},\\    &E_{\infty}=\{0,1\}.
 \end{split}
\end{equation*}

In the Cartesian product $E=E_\infty\times E_1 \times \cdots \times E_4$ we look
for such  elements $e$ for which 
\[
d(e):= e_\infty - \sum_{i=1}^4 e_i
\]
is a non-negative integer. We have two sets of exponents
\begin{equation}
  \label{eq:degP0}
\begin{split}
&e^{(1)}=\left\{1,\frac{3}{4},\frac{1}{4},\frac{1}{2},-\frac{1}{2} \right\},\quad d\left( e^{(1)}\right)=0,\\
&e^{(2)}=\left\{1,\frac{1}{4},\frac{3}{4},\frac{1}{2},-\frac{1}{2} \right\},\quad d\left( e^{(2)}\right)=0.
\end{split}
\end{equation}
For these two choices, we construct the function
$\omega=\sum_{i=1}^4\frac{e_i}{z-z_i}$, according to formula \eqref{eq:t1}, and check whether the condition of
existence of polynomial of degree $0$
\begin{equation}
\Dz \omega+\omega^2-r=0    
\label{eq:kova1deg0}
\end{equation}
is satisfied, compare with equation  \eqref{eq:P1}.  Easy calculations show that for both choices~\eqref{eq:degP0}
condition \eqref{eq:kova1deg0} is not fulfilled for positive moments of inertia. 

Now we consider the non-generic cases. Case $C=A$ is integrable. For $C=2A$ coefficient $r(z)$ takes the form
\begin{equation}
 \begin{split}
r(z)=&-\frac{3}{16 z^2}-\frac{3}{16 (z-1)^2}-\frac{1}{4 (2 A+z+1)^2}\\
&+\frac{3-2 A}{16 A z+8 z}+\frac{A+4}{8 (A+1) (z-1)}\\
&-\frac{10 A+7}{8 (A+1) (2 A+1) (2 A+z+1)}.
 \end{split}   
\end{equation}

With the above coefficient $r(z)$,  equation \eqref{eq:zredgeq0}  has
generically four singularities: three poles of the second order $z_1=0$,
$z_2=1$, $z_3=-(2 A+1)$ provided $(A+1)(2A+1)\neq0$ and $z_4=\infty$. The degree
of infinity is two, and the Laurent expansion at infinity is
\begin{equation*}
r(z)=\frac{3}{4 z^2}+O\left(\dfrac{1}{z^3}\right).
\end{equation*}
The coefficients $\alpha_i$ in the expansions \eqref{eq:roz} of $r(z)$ around
singularities are
\begin{gather*}
\alpha_1 =\alpha_2=-\frac{3}{16},\quad  \alpha_3=-\frac{1}{4},\quad \alpha_{\infty}=\frac{3}{4}
\end{gather*}
and the differences of exponents
\begin{equation*}
\Delta_1=\Delta_2 =\frac{1}{2},\quad  \Delta_3=0, \quad  \Delta_{\infty}=2.
\end{equation*}
Also in this case, logarithmic terms in local solutions of
\eqref{eq:zredgeq0} around $z_3$ are always present and the differential Galois groups
can be reducible or full $\mathrm{SL}(2,\C)$. The auxiliary sets of
exponents at singularities in the first case of the Kovacic algorithm are
\begin{equation*}
 E_1=E_2=\left\{\frac{1}{4},\frac{3}{4}\right\},\,\,\, E_3=\left\{\frac{1}{2},\frac{1}{2}\right\},\,\,\,    E_{\infty}=\left\{-\frac{1}{2},\frac{3}{2}\right\}.
\end{equation*}
In the Cartesian product $E=E_\infty\times E_1 \times \cdots \times E_3$ we look
for such  elements $e$ for which 
\[
d(e):= e_\infty - \sum_{i=1}^4 e_i,
\]
is a non-negative integer and we have two sets of exponents
\[
\begin{split}
&e^{(1)}=\left\{\frac{3}{2},\frac{3}{4},\frac{1}{4},\frac{1}{2} \right\},\quad d\left( e^{(1)}\right)=0,\\
&e^{(2)}=\left\{\frac{3}{2},\frac{1}{4},\frac{3}{4},\frac{1}{2} \right\},\quad d\left( e^{(2)}\right)=0,
\end{split}
\]
For these two choices we construct function
$\omega=\sum_{i=1}^4\frac{e_i}{z-z_i}$  condition of existence of polynomial of
degree 0  gives in these two cases conditions $A=0$ and $A+1=0$, respectively,
that are impossible and the first cases of the Kovacic algorithm is impossible.

For $C=A-1$ coefficient $r(z)$ simplifies to
\begin{equation*}
 \begin{split}
&r(z)=\frac{5}{16 z^2}-\frac{3}{16 (z-1)^2} -\frac{1}{4 \left(A^2-z\right)^2}+\frac{(A-2) A+2}{8 A^2 z}\\
&-\frac{(A-2) A+3}{8 \left(A^2-1\right) (z-1)}
+\frac{A^2-A+1}{4 (A-1) A^2 (A+1) \left(z-A^2\right)}
 \end{split}   
\end{equation*}
and equation \eqref{eq:zredgeq0} with this coefficient has generically four
singularities: three second-order poles $z_1=0$, $z_2=1$, $z_3=A^2$ provided
$A(A+1)(A-1)\neq0$ and $z_4=\infty$, and the degree of infinity is three
provided $A(A+1)\neq0$. The coefficients $\alpha_i$ in the expansions \eqref{eq:roz} of
$r(z)$ around singularities are
\begin{gather*}
\alpha_1 =\frac{5}{16},\quad \alpha_2=-\frac{3}{16}, \quad \alpha_3=-\frac{1}{4}
\end{gather*}
and the differences of exponents equal
\begin{equation*}
\Delta_1 =\frac{3}{2},\quad \Delta_2=\frac{1}{2},\quad \Delta_3=0.
\end{equation*}
As in previous cases, logarithmic terms in local solutions of \eqref{eq:zredgeq0}
around $z_3$ are present and the differential Galois groups can be
reducible or full $\mathrm{SL}(2,\C)$. The auxiliary sets of exponents at
singularities in the first case of the Kovacic algorithm are
\begin{equation*}
\begin{split}
 &E_1=\left\{-\frac{1}{4},\frac{5}{4}\right\},\,\,\, E_2=\left\{\frac{1}{4},\frac{3}{4}\right\},\,\,\, E_3=\left\{\frac{1}{2},\frac{1}{2}\right\},\\    &E_{\infty}=\{0,1\}.
 \end{split}
\end{equation*}
In the Cartesian product $E=E_\infty\times E_1 \times \cdots \times E_3$ we have
one element $e=\left(1,-\frac{1}{4},\frac{3}{4},\frac{1}{2}\right)$ for which
$d(e)=e_{\infty}-e_1-e_2-e_3=0$ is a nonnegative integer and we construct the
function $\omega=\sum_{i=1}^3\frac{e_i}{z-z_i}$. The condition
\eqref{eq:kova1deg0} of the existence of a polynomial of degree 0  gives $A (1 +
A)=0$ but these values are greater than zero. The special subcase $A=1$ of this case
gives $C=A-1=0$ and this case is excluded. In this way, we finished the proof of
Lemma~\ref{lem:3}.
\end{proof}
\begin{remark}
In case $g=d=0$   our system \eqref{eq:lag_borisovexpl} has two other invariant
manifolds: given by conditions $M_1=M_3=\Gamma_2=0$ (as for the case $g=0$ and
$d\neq0$) and another determined by  $M_1=M_2=\Gamma_3=0$ but reduced normal
variational equations along particular solutions lying on these manifolds have
differences of exponents depending on parameters $A$ and $C$ that makes analysis
more difficult.
\end{remark}

By Theorem~\ref{thm:Morales} using Theorem~\ref{thm:g0gen} and
Lemmata~\ref{lem:2} and \ref{lem:3} we obtain that for $g=0$ and $\vL\neq\vzero$
the only integrable cases are the Euler case and the Lagrange case. This result,
together with the results of the analysis in Section~\ref{sec:anal} for sliding
top in non-zero constant gravity field, finishes the proof of our main
Theorem~\ref{thm:we-main}.

\section{Final remarks}
\label{sec:finalremarks}

We perform a complete integrability analysis of the sliding top problem. The final
result was achieved by applying several methods. We show that
the considered system is, in a certain sense, a perturbation of either the heavy top or the Euler top. Using this fact, we show that the integrability of the classical top is the necessary condition of the sliding top. In effect, under the assumption that the gravity field does not vanish, our analysis was reduced to studying two cases of Kovlaevskaya and Goryachev-Chaplygin. 

The case of vanishing gravity appeared considerably harder. The basic problem in this case is the lack of a suitable particular solution that can be used to study the integrability. One way to overcome this difficulty was to add an appropriate
assumption concerning the problem parameters. However, this automatically limits
the generality of the results. This is why we decided to follow S.L. Ziglin Integrability Investigations of the classical top. The crucial point was to prove that if the sliding top is integrable, then the body is symmetric. We obtain this result using one, not so well-known, theorem of S.L. Ziglin describing the splitting of separatrices in a system with two degrees of freedom close to an integrable one.  

The similarities between the classical and sliding tops are misleading. Nevertheless, both tops have two common integrable cases of Euler and Lagrange. Moreover, we mentioned that for the sliding top we also have the Hess-Appelrot case.  It is remarkable that in this case, for the sliding top, similarly as for the classical
top, only one pair of surfaces asymptotic to the hyperbolic periodic solutions
intersects. One can prove it using the method of \cite{Ziglin:80::e} and our calculation presented in Section~\ref{sec:g0generic}.

It is known that the addition of gyrostat terms to the Euler equations of classical heavy top does not destroy the integrability in famous cases giving the   Zhukovskii case (an extension of the Euler case), the Lagrange case for gyrostat and  Yehia case (an extension of the Kovalevskaya case), see e.g.  \cite{Ollagnier:23::} and references therein. It is also integrable on the
level  $F_2=0$ t the Sretenskii case (an extension of the Goryachev-Chaplygin). One can ask whether the two integrable cases of the sliding top after the addition of the gyrostatic term remain integrable.

Equations of motion for the sliding top with gyrostat take the form
\begin{equation}
\begin{split}
 &\Dt\vM=(\vM-\vB)\times \vJ^{-1}\vM-mg\vL\times\vGamma\\
 &+ 
 m\left[\vGamma\cdot(\vJ^{-1}\vM\times\vL)\right]
 \left((\vJ^{-1}\vM\times\vL)\times\vGamma\right),\\
 &\Dt  \vGamma = \vGamma \times \vJ^{-1}\vM,
\end{split}
 \label{eq:lag_with_gyro} 
\end{equation}
where  $\vB$ is the constant gyrostatic moment.

System \eqref{eq:lag_with_gyro} has three first integrals 
\begin{equation*}
\begin{split}
&H=  \frac{1}{2} \vJ^{-1}\vM\cdot \vI \vJ^{-1}\vM+ 
\frac{1}{2} m \left[\vGamma\cdot( \vJ^{-1}\vM\times\vL)\right]^2 \\
&+ 
m g \vGamma\cdot\vL,\qquad F_1= \vGamma^2,\quad F_2=\vGamma\cdot(\vM-\vB).
\end{split}
\label{integralsMM} 
\end{equation*}

Equations of motion \eqref{eq:lag_with_gyro} can be written as Hamiltonian system \eqref{eq:poissonowska} with degenerated 
Poisson structure defined by slightly modified matrix $\J(\vX)$ 
\begin{equation*}
  \J(\vX)= \begin{bmatrix}
    \widehat\vM- \widehat\vB &  \widehat\vGamma \\
    \widehat\vGamma & \vzero
  \end{bmatrix}.
\end{equation*}
The functions $F_1$ and $F_2$ are Casimir functions of this Poisson structure.

We have additional first integral  in cases that are generalizations of the Euler and the Lagrange rigid body cases:
\begin{enumerate}
    \item  in the  Zhukovskii--Volterra
 case when $L_1=L_2=L_3=0$ that is 
\[
F_3=\vM\cdot\vM-2\vB\cdot\vM
\]
and 
\item in the Lagrange case when $B=A$ and $L_1=L_2=B_1=B_2=0$ equals to $F_3=M_3$. In the completely symmetric case when $A=B=C$ the additional first integral is $F_3=\vM\cdot\vL$ provided $\vB\times\vL=\vzero$.
\end{enumerate}
 In fact change of variables $\vM - \vB \to \vM$ transforms the first case into the Zhukovskii--Volterra system with the standard bracket defined by the Lie algebra $\mathfrak{e}(3)$, for details see \cite[ Sec. 2.7, Ch. 2]{Borisov:19::}.

\section*{Acknowledgements}
The authors wish to express their thanks to Ivan S. Mamaev for drawing attention to the equivalence of equations of motion for a sliding top and a ball with a displaced center of mass moving on a perfectly smooth horizontal plane.

This research was founded by The
National Science Center of Poland under Grant No.
2020/39/D/ST1/01632.

\section*{Data Availability Statement}
Data sharing is not applicable to this article as no new data were created or analyzed in this study.

\appendix
\section{Differential Galois group of second order differential equation with rational coefficients}
\label{sec:2ndorderlde}
Let us consider a second-order differential equation of the following form
\begin{equation}
\label{eq:gso}
 y''=r(z) y, \qquad r(z)\in\C(z), \qquad '\equiv \frac{\rmd\phantom{z}}{\rmd z} .
\end{equation}
For this equation its differential Galois group $\scG$ is an algebraic
subgroup of $\mathrm{SL}(2,\C)$. The following lemma describes all
possible types of $\scG$ and relates these types to forms of solution
of \eqref{eq:gso}, see \cite{Kovacic:86::,Morales:99::c}.
\begin{lemma}
\label{lem:alg}
Let $\scG$ be the differential Galois group of equation~\eqref{eq:gso}.
Then one of four cases can occur.
\begin{enumerate}
\item $\scG$ is reducible (it is conjugated to a subgroup of the triangular
  group); in this case, equation \eqref{eq:gso} has an exponential
  solution of the form $y=\exp\int \omega$, where $\omega\in\C(z)$,
\item $\scG$ is conjugated with a subgroup of
\begin{equation}
D^\dag = \left\{ \begin{bmatrix} c & 0\\
                                0 & c^{-1}
                      \end{bmatrix}  \; \biggl| \; c\in\C^*\right\} \cup
                      \left\{ \begin{bmatrix} 0 & c\\
                                c^{-1} & 0
                      \end{bmatrix}  \; \biggl| \; c\in\C^*\right\},
\label{eq:Ddag}                      
\end{equation}
  in this case equation
  \eqref{eq:gso} has a solution of the form $y=\exp\int \omega$, where
  $\omega$ is algebraic over $\C(z)$ of degree 2,
\item $\scG$ is primitive and finite; in this case all
  solutions of equation \eqref{eq:gso} are algebraic,

\item $\scG= \mathrm{SL}(2,\C)$ and equation \eqref{eq:gso}
  has no Liouvillian solution.
\end{enumerate}
\end{lemma}

We need a more precise characterization of case 1 in the above
lemma. It is given by the following lemma, see Lemma~4.2 in
\cite{Singer:93::a}.
\begin{lemma}
\label{lem:algc1}
  Let $\scG$ be the differential Galois group of
  equation~\eqref{eq:gso} and assume that $\scG$ is reducible.
  Then either
\begin{enumerate}
\item equation~\eqref{eq:gso} has a unique solution $y$ such that
  $y'/y\in\C(z)$, and $\scG$ is conjugate to a subgroup of the
  triangular group
\begin{equation}
 \scT = \left\{ \begin{bmatrix} a & b\\
                        0& a^{-1}
               \end{bmatrix} \, |\, a,b\in\C, a\neq 0\right\}.
               \label{eq:Tfull}
\end{equation}
Moreover, $\scG$ is a proper subgroup of $\scT$ if and only if there
exists $m\in\N$ such that $y^m\in\C(z)$. In this case $\scG$ is
conjugate to
\begin{equation}
 \scT_m = \left\{ \begin{bmatrix} a & b\\
                        0& a^{-1}
               \end{bmatrix} \, |\, a,b\in\C, a^m=1\right\},
\end{equation}
where $m$ is the smallest positive integer such that $y^m\in\C(z)$, or
\item equation~\eqref{eq:gso} has two linearly independent solutions
  $y_1$ and $y_2$ such that $y'_i/y_i\in\C(z)$, then $\scG$ is
  conjugate to a subgroup of
\begin{equation}
 \scD = \left\{ \begin{bmatrix} a & 0\\
                        0& a^{-1}
               \end{bmatrix} \, |\, a\in\C, a\neq 0\right\}.
\end{equation}
In this case, $y_1y_2\in\C(z)$. Furthermore, $\scG$ is conjugate to a
proper subgroup of $\scD$ if and only if $y_1^m\in\C(z)$ for some
$m\in\N$. In this case $\scG$ is a cyclic group of order $m$ where $m$
is the smallest positive integer such that $y_1^m\in\C(z)$.
\end{enumerate}
\end{lemma}

In case 2 of the above lemma we know that $v=y_1y_2\in\C(z)$.
Differentiating $v$ three times, and  using the fact that $y_i$ satisfies
equation \eqref{eq:gso},  we obtain
\begin{equation}
\label{eq:ssp}
 v'''= 2r' v + 4r v'.
\end{equation}
The above equation is called the second symmetric power of
equation~\eqref{eq:gso}. For applications of symmetric powers of
differential operators to study the existence of Liouvillian
solutions and differential Galois group,  see e.g.
\cite{Singer:93::a,Singer:95::,Ulmer:96::}.

To decide if case 2 from Lemma~\ref{lem:alg} occurs we can apply the
Kovacic algorithm \cite{Kovacic:86::}. The algorithm consists of four cases that correspond exactly to the cases listed in the Lemma~\ref{lem:alg} and is necessary to test the second case.

The presence of logarithms in local solutions of~\eqref{eq:gso} around a singularity also gives strong restrictions on the differential Galois group of this equation. Then only case 1 or case 4 of Lemma~\ref{lem:alg} can happen. 
 
Here, we present the first and the second cases of the Kovacic algorithm corresponding to differential Galois groups described in items 1 and 2 of Lemma~\ref{lem:alg}, and we restrict ourselves to the Fuchsian linear differential equations.  First, we introduce the notation.  We write $r(z)\in\mathbb{C}(z)$ in the form
\begin{equation*}
\label{eq:rst}
r(z) = \frac{s(z)}{t(z)}, \qquad s(z),\, t(z) \in \mathbb{C}[z],
\end{equation*}
where $s(z)$ and $t(z)$ are relatively prime polynomials and $t(z)$ is monic.  The roots of $t(z)$ are poles of $r(z)$. We denote $\Sigma:=
\{ c\in\mathbb{C}\,\vert\, t(c) =0 \}$.  The order $\ord(c)$ of $c\in\Sigma$
is equal to the multiplicity of $c$ as a root of $t(z)$. The infinity is the singularity of equation \eqref{eq:gso} and the order of infinity is defined by
\[
\mathrm{ord}(\infty):=\deg t- \deg s.
\]

Because we assume that equation~\eqref{eq:gso} is Fuchsian,  we have
$\ord(c)\leq 2$ for $c\in\Sigma$ and $\mathrm{ord}(\infty)\geq2$.

For each $c\in\Sigma$ we have the
following expansion
\begin{equation}
r(z) = \frac{a_c}{(z-c)^2} + O\left( \frac{1}{z-c}\right),
\label{eq:r}
\end{equation}
and we define $\Delta_c = \sqrt{1+4a_c}$. For infinity, we have

\begin{equation}
r(z)=\dfrac{a_\infty}{z^2}+O\left(\dfrac{1}{z^3}\right),
\label{eq:Linfty}
\end{equation}
and we define $\Delta_\infty = \sqrt{1+4a_\infty}$.

\noindent
\textsc{Case I}

\noindent
\textbf{Step I.} 
For each $c\in\Gamma\cup\{\infty\}$ we
define  two complex numbers
$\alpha_c^{+}$, $\alpha_c^{-}$ as described below.
\begin{description}
\item[(c$_1$)] If $c\in\Gamma$ and $\ord(c)=1,$ then 
\[
\alpha_c^{+}=\alpha_c^{-}=1.
\]
\item[(c$_2$)] If $c\in\Gamma$ and $\ord(c)=2,$ and $r$ has the  expansion of the form \eqref{eq:r}, then
\begin{equation}
\alpha_c^{\pm}=  \frac{1}{2}\left(1\pm\Delta_c\right)
\label{eq:del1s}
\end{equation}
\item[($\infty_1$)] If $\mathrm{ord}(\infty)>2$, then
\[
 \alpha_{\infty}^{+}=0,\quad \alpha_{\infty}^{-}=1.
\]
\item[($\infty_2$)] If  $\mathrm{ord}(\infty)=2$ and  the Laurent series expansion of $r$ at $\infty$ takes the form \eqref{eq:Linfty},
then
\begin{equation}
\alpha_{\infty}^{\pm}=  \frac{1}{2}\left(1\pm\Delta_\infty\right).
\label{eq:delinf}
\end{equation}
\end{description}
\noindent
\textbf{Step II.} For each family $s=(s(c),s(\infty))$, $c\in\Gamma$, where
$s(c)$ and $s(\infty)$ are either $+$ either $-$, let
we compute
\[
  d := e_\infty^{s(\infty)}- \sum_{c\in\Gamma}\alpha_c^{s(c)}.
\]
If $d$ is a non-negative integer, then
\begin{equation}
 \omega(z) =
\sum_{c\in\Gamma}\frac{\alpha_c^{s(c)}}{z-c},
\label{eq:t1}
\end{equation}
is a candidate for $\omega$.  If there are no such elements,
equation~\eqref{eq:gso} does
not have an exponential solution and the algorithm stops here.

\noindent
\textbf{Step III.} For each family from step II that gives 
$d\in\N_0$ we  search for a monic polynomial $P=P(z)$
of degree $d$ satisfying the
following equation
\begin{equation}
\label{eq:P1}
P'' + 2\omega(z)P' +( \omega'(z)+ \omega(z)^2 -r(z))P =0.
\end{equation}
If such polynomial exists, then equation~\eqref{eq:gso} possesses an
exponential solution of the form $y=P\exp\int\theta$, where $\theta=\omega$,
if not,
equation~\eqref{eq:gso} does not have an exponential solution.
\\[\bigskipamount]
\noindent
\textsc{Case II}

\noindent
\textbf{Step I.}
For each $c\in\Gamma\cup\{\infty\}$ we define   $E_c$ as described below.
\begin{description}
\item[(c$_1$)] If $c\in\Gamma$ and $\ord(c)=1,$ then the set $E_{c}=\{4\}$.
\item[(c$_2$)] If $c\in\Gamma,$ $\ord(c)=2$ and $r$ has the  expansion of the form
\eqref{eq:r}, then
\begin{equation}
  E_c := \left\{ 2 , 2\left(1+\Delta_c\right),
   2\left(1-\Delta_c\right)\right\}\cap\Z,
\label{eq:del2s}
\end{equation}
\item[($\infty_{1}$)] If $\ord{\infty}>2$, then $E_{\infty}=\{1,2,4\}$
\item[($\infty_{2}$)] If $\ord{\infty}=2$ and the he Laurent series expansion of
$r$ at $\infty$ takes the form \eqref{eq:Linfty}, then
\begin{equation}
  E_\infty := \left\{2,  2\left(1+\Delta_{\infty}\right),
  2\left(1-\Delta_{\infty}\right)\right\}\cap\Z.
\label{eq:del2i}
\end{equation}
\end{description}

\noindent
\textbf{Step II.} We consider all families $(e_c)_{c\in\Gamma\cup\{\infty\}}$
with $e_c\in E_c$ and at least one of the coordinates is odd. 
We compute
\[
  d := \frac{1}{2}\left(e_\infty- \sum_{c\in\Gamma}e_c\right).
\]
and  select those  families $(e_c)_{c\in\Gamma\cup\{\infty\}}$ for which $d(e)$
is a non-negative
integer.  If there are no such elements,  Case II cannot occur and the
algorithm stops here.

\noindent
\textbf{Step III.} For each family  giving
$d\in\N_0$ we define
\begin{equation}
\omega=\omega(z) = \frac{1}{2}
\sum_{c\in\Gamma}\frac{e_c}{z-c},
\label{eq:t2}
\end{equation}
and we search for a monic polynomial $P=P(z)$ of degree $d$ satisfying the
following equation
\begin{equation}
\begin{split}
&P''' + 3\omega P'' +(3 \omega^2 + 3\omega' -4r)P' \\
&+
(\omega'' + 3 \omega\omega' + \omega^3 -4 r\omega - 2r')P =0.
\end{split}
\label{eq:P2}
\end{equation}
If such a polynomial exists, then equation~\eqref{eq:gso} possesses a
 solution of the form $y=\exp\int\theta$, where
\[
\theta^2 - \psi\theta +\frac{1}{2}\psi' + \frac{1}{2}\psi^2 - r =0, \qquad
\psi = \omega + \frac{P'}{P}.
\]
If we do not find such a polynomial, then Case II in Lemma~\ref{lem:alg}
cannot occur.

\section{Differential Galois group of the Lam\'e equation}
\label{sec:lamiaste}

The Weierstrass 
form of the Lam\'e equation is the following
\begin{equation} 
\label{eq:lame2} 
\frac{\mathrm{d}^2y}{\mathrm{d}t^2}=(\alpha\wp(t;g_2,g_3)+\beta)y, 
\end{equation} 
where $\alpha$ and $\beta$ are, in general, complex parameters and
$\wp(t;g_2,g_3)$ is the elliptic Weierstrass function with invariants
$g_2$, $g_3$. In other words, $\wp(t;g_2,g_3)$ is a solution of the
differential equation
\begin{equation} 
\dot v^2=f(v),\qquad f(v)=4v^3-g_2v-g_3. 
\label{wei} 
\end{equation} 
The polynomial $f(v)$ is assumed to have three different roots, so
\begin{equation*} 
\Delta=g_2^3 -27g_3^2 \neq  0.
\end{equation*} 
The modular function $j(g_2,g_3)$ associated with the
elliptic curve \eqref{wei}  is defined as follows
\begin{equation*}
\label{eq:j}
j(g_2,g_3)=\frac{g_2^2}{g_2^3 -27g_3^2}.
\end{equation*}
The algebraic form of the Lam\'e equation is
\begin{equation}
\frac{\mathrm{d}^2y}{\mathrm{d} x^2} +\frac{f'(x)}{2f(x)} \frac{\mathrm{d}y}{\mathrm{d} x}-\frac{\alpha x+\beta}{f(x)}m=0,
\label{eq:lamaalga}
\end{equation}
and is related to the Weierstrass form \eqref{eq:lame2} by the transformation $x=\wp(t)$.

Classically, the Lam'e equation is written with parameter $n$
instead of $\alpha$ which are related by the formula $\alpha=n(n+1)$.  

We see that the Lam\'e equation depends on four parameters $(n,\beta,g_2,g_3)$.
The following lemma lists all the cases in which the identity
component of the differential Galois group of Lam\'e
equation~\eqref{eq:lame2} is Abelian, see, e.g.,
\cite[Sec.~2.8.4]{Morales:99::c} and references therein.
\begin{lemma} 
\label{lem:lame}
The identity component of the differential Galois group of Lam\'e
equation~\eqref{eq:lame2} is Abelian only in the following cases:
\begin{enumerate} 
\item the Lam\'e-Hermite case when $n\in\mathbb{Z}$ i.e.  $\delta:=n+\frac{1}{2}\in \frac{1}{2}\mathbb{Z}$ and three other parameters are arbitrary.
\item the Brioschi-Halphen-Crowford case for which
$\delta:=n+\frac{1}{2}\in\mathbb{N}$, and the remaining parameters $(g_2,g_3,\beta)$ satisfy an algebraic equation obtained from the condition of vanishing of the so-called  Brioschi determinant.
\item the Baldassarri case for which
  $\delta:=n+\frac{1}{2}\in \frac{1}{3}\mathbb{Z}\cup  \frac{1}{4}\mathbb{Z}\cup \frac{1}{5}\mathbb{Z}\setminus\mathbb{Z}$,
 with  additional algebraic restrictions on $(g_2,g_3,\beta)$. 
\end{enumerate} 
\end{lemma}
\section{Formulae for proof of Lemma~\ref{lem:3}}
\label{sec:residius_explicitely}

Function $h_1$ on the first double-asymptotic  solution, with  sign $+$ in equation \eqref{eq:solsEulertop},
takes the form
\begin{widetext}
\begin{equation}
  \label{eq:h1z}
\begin{split}
& h_1^{+}(z)=-\frac{2m \alpha  \beta \aH \rme^{\rmi g_0} (\alpha -\rmi \beta ) \sqrt{G_0^2-\aH^2}
   z^{1+\frac{\rmi b}{\alpha  \beta }}}{(z^2+1)^4 (\alpha ^2+\beta
   ^2)^{3/2}}    \Big[2 \alpha  \beta  z \Big(-S_2 (z^4-6 z^2+1)
   \sqrt{\alpha ^2+\beta ^2} (\beta  S_1+\alpha  S_3)\\
  & +2 z (z^2-1) (\beta S_1+\alpha  S_3)^2-2 S_2^2 z(z^2-1) (\alpha ^2+\beta
   ^2)\Big)+b (z^2+1)  \Big(\beta  S_1^2 
   \left(\alpha  (z^4-1)+\rmi\beta  (z^2-1)^2\right)\\
   &+S_1 \left[2 S_2 z \sqrt{\alpha ^2+\beta ^2}
   \left(\alpha  (z^2+1)+2 \rmi \beta(z^2-1)\right)+S_3 \left((\alpha^2-\beta^2)
   (z^4-1)+2 \rmi \alpha  \beta 
   (z^2-1)^2\right)\right]\\
   &+4 \rmi  (\alpha ^2+\beta ^2) S_2^2 z^2+2 \rmi S_2
   S_3 z \sqrt{\alpha ^2+\beta ^2} \left(2 \alpha (z^2-1)+\rmi \beta 
   (z^2+1)\right)+\alpha  S_3^2 
   \big[\beta(1 -z^4)
   +\rmi \alpha 
   (z^2-1)^2\big]\Big)\Big],
\end{split}    
\end{equation}
\end{widetext}
and function $h_2$ on it is the following
\begin{widetext}
\begin{equation}
  \label{eq:h2z}
\begin{split}
&h_2^{+}(z)=\frac{\rme^{2 \rmi g_0} m (\alpha -\rmi \beta )^2 (\aH^2-G_0^2)
    z^{\frac{2 \rmi b}{\alpha  \beta }}}{4
   (z^2+1)^4 (\alpha^2+\beta^2)^2}\Big[ b^2 (z^2+1)^2 \Big(2 S_1 
   \left(\beta +\rmi \alpha (z^2+1)-\beta  z^2\right) \Big[2 S_2 z \sqrt{\alpha^2+\beta ^2}+\alpha  S_3 (z^2-1)\\
   &+\rmi \beta  S_3
   (z^2+1)\Big]+\left(\alpha  S_1
   (z^2+1)+\rmi \beta  S_1 (z^2-1)\right)^2-4
   S_2^2 z^2 (\alpha^2+\beta^2) -4 S_2 S_3
z \sqrt{\alpha ^2+\beta ^2}\left(\alpha  (z^2-1)+\rmi \beta 
   (z^2+1)\right)\\
   &-\left(\alpha  S_3
   (z^2-1)+\rmi \beta  S_3
   (z^2+1)\right)^2\Big)-4 \alpha  b \beta  z
   (z^2+1)   \Big[S_2 \sqrt{\alpha ^2+\beta ^2}
   \Big((z^4-1) (\alpha  S_1-\beta  S_3)+\rmi
   (z^4-6 z^2+1) (\beta  S_1\\
   &+\alpha 
   S_3)\Big)+2 z (\beta  S_1+\alpha  S_3)
   \left((z^2+1) (\beta  S_3-\alpha  S_1)-\rmi
   (z^2-1) (\beta  S_1+\alpha  S_3)\right)+2 \rmi
   S_2^2 z (z^2-1) \left(\alpha ^2+\beta
   ^2\right)\Big] \\
   &+4 \alpha ^2 \beta ^2 z^2 \Big(S_2
  (z^2-1)\sqrt{\alpha ^2+\beta ^2} \Big[S_2
   (z^2-1) \sqrt{\alpha ^2+\beta ^2}-4 z (\beta 
   S_1+\alpha  S_3)\Big]+4 z^2 (\beta  S_1+\alpha 
   S_3)^2\Big)
   \Big],
\end{split}   
\end{equation}
\end{widetext}
where $z=\rme^{\lambda t}$, $\alpha=\sqrt{b-a}$ and $\beta=\sqrt{c-b}$. From the forms of these functions, it follows that they can have residues only at $z=\pm\rmi$. 

The sum of residues of function $h_1^{+}(z)$  equals
\begin{widetext}
\begin{equation}
\begin{split}
&\sum_{i=1}^2  \operatorname{res}\,(h_1^{+},z_i) = \frac{mb \aH (\alpha -\rmi \beta ) \sqrt{G_0^2-\aH^2} \rme^{-\frac{\pi  b}{2 \alpha  \beta }+\rmi g_0}}{6 \alpha  \beta  \left(\alpha ^2+\beta
   ^2\right)^{3/2}}\Big\{\left(e^{\frac{\pi  b}{\alpha  \beta }}+1\right) \Big[2 \rmi b^2 \left((\beta  S_1+\alpha 
   S_3)^2-S_2^2 \left(\alpha^2+\beta^2\right)\right)\\
 &  +3 \alpha  \beta  b \left(-S_2^2
   \left(\alpha ^2+\beta ^2\right)+(\beta  S_1+\alpha  S_3) (S_1 (\beta -2 \rmi \alpha
   )+S_3 (\alpha +2 \rmi \beta ))\right)-2 \alpha ^2 \beta ^2 \Big((S_1 (3 \alpha +2 \rmi \beta
   )
  +S_3 (-3 \beta +2 \rmi \alpha )) \\
  &\cdot(\beta  S_1+\alpha  S_3)+\rmi S_2^2 \left(\alpha ^2+\beta
   ^2\right)\Big)\Big] +2\sqrt{\alpha ^2+\beta
   ^2} S_2 \left(1-e^{\frac{\pi  b}{\alpha  \beta }}\right) (b-\rmi \alpha  \beta ) \Big[2 b (\beta  S_1
   +\alpha 
   S_3)-\alpha  \beta  (S_1(3 \alpha +\rmi \beta)\\
   &+S_3(\rmi \alpha-3 \beta))\Big] \Big\},
   \end{split}
\label{eq:resh1+}  
\end{equation}
\end{widetext}
where $z_1=-z_2=\rmi$.

For function $h_2^+(z)$  the respective sum of residues reads
\begin{widetext}
\begin{equation}
\begin{split}
&\sum_{i=1}^2  \operatorname{res}\,(h_2^{+},z_i) =
\frac{m (G_0^2-\aH^2) (\alpha -\rmi \beta )^2\rme^{-\frac{\pi  b}{\alpha  \beta }+2
   \rmi g_0}}{24 \alpha  \beta  (\alpha
   ^2+\beta ^2)^2}\Big\{ \left(e^{\frac{2 \pi  b}{\alpha  \beta }}-1\right) \Big[4 b^3 \Big(S_2^2 (\alpha^2+\beta^2)-(\beta  S_1+\alpha  S_3)^2\Big)\\
&   +6 \alpha  \beta  b^2 \left((S_1 (2
   \alpha +\rmi \beta )+\rmi S_3 (\alpha +2 \rmi \beta )) (\beta  S_1+\alpha  S_3)-\rmi S_2^2
   \left(\alpha ^2+\beta ^2\right)\right)+4 \alpha ^2 \beta^2 b \Big(S_2^2 \left(\alpha ^2+\beta
   ^2\right)
   +(\beta  S_1+\alpha  S_3)\\ &\cdot (S_1 (2 \beta -3 \rmi \alpha )+S_3 (2 \alpha
   +3 \rmi \beta ))\Big)-3 \rmi \alpha ^3 \beta ^3 \left((\beta  S_1+\alpha  S_3)^2+S_2^2
   \left(\alpha ^2+\beta ^2\right)\right)\Big]
-4 \rmi\sqrt{\alpha
   ^2+\beta^2} b S_2 \left(e^{\frac{2 \pi  b}{\alpha  \beta }}+1\right) \\ &\cdot (b-\rmi \alpha  \beta )\Big[2 b (\beta 
   S_1+\alpha  S_3)-\alpha  \beta  (S_1 (3 \alpha +\rmi \beta )+\rmi S_3 (\alpha +3 \rmi \beta ))\Big]
   \Big\}.
\end{split}
 \label{eq:resh2+}  
\end{equation}
\end{widetext}

Functions $h_1$ and $h_2$ on the second double-asymptotic solution, that with the sign  minus in formula~\eqref{eq:solsEulertop},  take the forms
\begin{widetext}
\begin{equation}
\begin{split}
&h_1^{-}(z)=-\frac{2m \alpha  \beta \aH \rme^{\rmi g_0} (\alpha -\rmi \beta ) \sqrt{G_0^2-\aH^2}
   z^{1+\frac{\rmi b}{\alpha  \beta }}}{(z^2+1)^4 (\alpha ^2+\beta
   ^2)^{3/2}}    \Big\{2 \alpha  \beta  z \Big(S_2 \left(z^4-6 z^2+1\right) \sqrt{\alpha ^2+\beta ^2} (\alpha 
   S_3-\beta  S_1)
   -2 z (z^2-1) \\ &\cdot (\alpha  S_3-\beta  S_1)^2+2
   S_2^2 z (z^2-1) (\alpha ^2+\beta ^2)\Big)-b (z^2+1)
   \Big(\beta  S_1^2 \left(\alpha  \left(z^4-1\right)+\rmi \beta 
   (z^2-1)^2\right)
   +S_1 \Big(S_3 \Big(-\alpha ^2
   \left(z^4-1\right)\\ &+\beta ^2 \left(z^4-1\right)-2 \rmi \alpha  \beta 
   (z^2-1)^2\Big)-2 S_2 z \sqrt{\alpha ^2+\beta ^2} \left(\alpha 
   (z^2+1)+2 \rmi \beta  (z^2-1)\right)\Big)
   +4 \rmi S_2^2 z^2
   (\alpha ^2+\beta ^2)\\ &+2 \rmi S_2 S_3 z \sqrt{\alpha ^2+\beta ^2} \left(2
   \alpha  (z^2-1)+\rmi \beta  (z^2+1)\right)+\alpha  S_3^2 \left(\beta
   -\beta  z^4+\rmi \alpha  (z^2-1)^2\right)\Big)
   \Big\},
\end{split}
\label{eq:h1zm}
\end{equation}
\end{widetext}
\begin{widetext}
\begin{equation}
\begin{split}
&h_2^{-}(z)=\frac{\rme^{2 \rmi g_0} m (\alpha -\rmi \beta )^2 (\aH^2-G_0^2)
    z^{\frac{2 \rmi b}{\alpha  \beta }}}{4
   (z^2+1)^4 (\alpha^2+\beta^2)^2}\Big\{
b^2 (z^2+1)^2 \Big[2 S_1 \left(\alpha 
   (z^2+1)+\rmi \beta  (z^2-1)\right) \Big[-2
   \rmi S_2 z \sqrt{\alpha ^2+\beta ^2}
  \\& -\rmi \alpha  S_3
  (z^2-1)+\beta  S_3
   (z^2+1)\Big]+\left(\alpha  S_1
   (z^2+1)+\rmi \beta  S_1
  (z^2-1)\right)^2-4 S_2^2 z^2 \left(\alpha
   ^2+\beta ^2\right)-4 S_2 S_3
    z \sqrt{\alpha^2+\beta^2}\\&\cdot \left(\alpha  (z^2-1)+\rmi \beta 
   (z^2+1)\right)-\left(\alpha  S_3
   (z^2-1)+\rmi \beta  S_3
   (z^2+1)\right)^2\Big]+4 \alpha  b \beta  z
  (z^2+1)\Big[S_2 \sqrt{\alpha^2+\beta^2}
   \Big((z^4-1) (\alpha  S_1\\ &+\beta 
   S_3)-\rmi \left(z^4-6 z^2+1\right) (\alpha  S_3-\beta
    S_1)\Big)+2 \rmi z (\alpha  S_3-\beta  S_1)
    \left((z^2-1) (\alpha  S_3-\beta 
   S_1)+\rmi (z^2+1) (\alpha  S_1+\beta 
   S_3)\right)\\ &-2 \rmi S_2^2 z (z^2-1)
   \left(\alpha^2+\beta^2\right)\Big]
   +4 \alpha^2 \beta^2
   z^2 \Big(4 z (\alpha  S_3-\beta  S_1) \left(z
   (\alpha  S_3-\beta  S_1)-S_2
   (z^2-1) \sqrt{\alpha ^2+\beta
   ^2}\right)\\&+S_2^2 (z^2-1)^2 (\alpha
   ^2+\beta ^2)\Big)   
\Big\},
\end{split}   
\label{eq:h2zm}
\end{equation}
\end{widetext}
respectively. The  sums of residues of $h_1^{-}(z)$ and $h_2^{-}(z)$ equal
\begin{widetext}
\begin{equation}
\begin{split}
&\sum_{i=1}^2  \operatorname{res}\,(h_1^{-},z_i) = \frac{mb \aH (\alpha -\rmi \beta ) \sqrt{G_0^2-\aH^2} \rme^{-\frac{\pi  b}{2 \alpha  \beta }+\rmi g_0}}{6 \alpha  \beta  (\alpha ^2+\beta^2)^{3/2}}\Big\{ \left(\rme^{\frac{\pi  b}{\alpha  \beta }}+1\right) \Big[2 \rmi b^2
   \left(S_2^2 (\alpha ^2+\beta ^2)-(\alpha 
   S_3-\beta  S_1)^2\right)\\
  & +3 \alpha  \beta  b
   \Big(S_2^2 (\alpha ^2+\beta ^2)+(\alpha 
   S_3-\beta  S_1) (S_1 (\beta -2 \rmi \alpha )-S_3
   (\alpha +2 \rmi \beta ))\Big)
   -2 \alpha ^2 \beta ^2 \Big(
    (S_1 (3 \alpha +2 \rmi \beta)
    +S_3 (3 \beta -2 \rmi \alpha ))\\
    &\cdot(\alpha  S_3-\beta  S_1)-\rmi S_2^2 \left(\alpha
   ^2+\beta ^2\right)\Big)\Big]+2 S_2 \sqrt{\alpha ^2+\beta
   ^2} \left(\rme^{\frac{\pi  b}{\alpha  \beta }}-1\right) (b-\rmi \alpha 
   \beta ) \Big[2 b (\alpha  S_3
   -\beta  S_1)+\alpha  \beta 
   (S_1 (3 \alpha +\rmi \beta )\\
   &+S_3 (3 \beta -\rmi \alpha ))\Big]\Big\},
   \end{split}
\label{eq:resh1-}
\end{equation}
\end{widetext}
\begin{widetext}
\begin{equation}
\begin{split}
&\sum_{i=1}^2  \operatorname{res}\,(h_2^{-},z_i) =
\frac{m (G_0^2-\aH^2) (\alpha -\rmi \beta )^2\rme^{-\frac{\pi  b}{\alpha  \beta }+2
   \rmi g_0}}{24 \alpha  \beta  (\alpha
   ^2+\beta ^2)^2}\Big\{ \left(\rme^{\frac{2 \pi  b}{\alpha  \beta }}-1\right) \Big[4 b^3
   \left(S_2^2 \left(\alpha ^2+\beta ^2\right)-(\alpha 
   S_3-\beta  S_1)^2\right)\\
   &+6 \alpha  \beta  b^2
   \left((\beta  S_1-\alpha  S_3) (S_1 (2
   \alpha +\rmi \beta )+S_3 (2 \beta -\rmi \alpha ))-\rmi
   S_2^2 \left(\alpha ^2+\beta ^2\right)\right)+4 \alpha ^2
   \beta ^2 b \Big(S_2^2 \left(\alpha ^2+\beta
   ^2\right)
   +(S_1 (-2 \beta \\ &+3 \rmi \alpha )+S_3 (2
   \alpha +3 \rmi \beta )) (\alpha  S_3-\beta 
   S_1)\Big)-3 \rmi \alpha ^3 \beta ^3 \left((\alpha 
   S_3-\beta  S_1)^2+S_2^2 \left(\alpha
   ^2+\beta ^2\right)\right)\Big]
   +4 b S_2 \sqrt{\alpha
   ^2+\beta ^2} \left(\rme^{\frac{2 \pi  b}{\alpha  \beta
   }}+1\right) \\ &(b-\rmi \alpha  \beta ) (\alpha  \beta  (S_1
   (\beta -3 \rmi \alpha )-S_3 (\alpha +3 \rmi \beta ))-2 \rmi b
   (\alpha  S_3-\beta  S_1))
\Big\},
\end{split}
\label{eq:resh2-}
\end{equation}  
\end{widetext}
respectively. These formulas are used in section~\ref{sec:g0generic}.
\bibliographystyle{unsrtnat}
\newcommand{\noopsort}[1]{}\def\cprime{$'$}
  \def\cydot{\leavevmode\raise.4ex\hbox{.}} \def\cprime{$'$} \def\cprime{$'$}
  \def\polhk#1{\setbox0=\hbox{#1}{\ooalign{\hidewidth
  \lower1.5ex\hbox{`}\hidewidth\crcr\unhbox0}}}

\end{document}